\documentclass[a4paper]{article}
\usepackage[utf8]{inputenc}
\usepackage{times}
\usepackage{soul}
\usepackage{url}
\usepackage{booktabs}
\usepackage{algorithm}
\usepackage{algorithmic}
\usepackage{comment}
\usepackage{csquotes}
\usepackage{ulem}
\usepackage{authblk}

\usepackage{amsmath}
\usepackage{amsthm}
\usepackage{amssymb}
\usepackage{amsfonts}
\usepackage{mathtools}

\usepackage{graphicx}
\usepackage{xcolor}
\usepackage{csquotes}
\usepackage{stmaryrd}
\usepackage{enumerate}
\usepackage[colorinlistoftodos]{todonotes}

\usepackage{bbm}
\usepackage{url}
\usepackage{wasysym}
\usepackage{dsfont}
\usepackage{wasysym}
\usepackage{bussproofs}
\EnableBpAbbreviations

\usepackage{tikz}
\usetikzlibrary{shapes,arrows}

\usepackage{graphicx}

\usepackage{xcolor}
\usepackage{cancel}
\usepackage{txfonts}

\usepackage{multicol}

\newcommand{\F}{\mathbb{F}}
\newcommand{\Plogic}{\mathbb{P}}
\newcommand{\KDf}{\mathbb{KD}4}
\newcommand{\KDff}{\mathbb{KD}45}

\newcommand{\Sf}{\mathbb{S}5}

\newcommand{\M}{\mathcal{M}}

\newcommand{\ForCL}{\mathcal{L}_{CL}}
\newcommand{\ForT}{\mathcal{L}_T}

\newcommand{\supp}{\rightsquigarrow}
\definecolor{rosa}{rgb}{1.0, 0.41, 0.71}
\newcommand{\NonMonot}{|\!\!\!\sim}
\def\IFF{\mbox{$\leftrightarrow$}}

\newtheorem{fact}{Fact}

\newcommand{\logic}{SBTrust}
\newtheorem{example}{Example}
\newtheorem{remark}{Remark}
\newtheorem{definition}{Definition}
\newtheorem{lemma}{Lemma}
\newtheorem{theorem}{Theorem}
\newtheorem{corollary}{Corollary}

\title{Support + Belief = Decision Trust} %TODO Please add

\author[1]{Alessandro Aldini}
\author[2]{Agata Ciabattoni}
\author[2]{Dominik Pichler}
\author[1]{Mirko Tagliaferri\footnote{Corresponding author. Email: mirko.tagliaferri@uniurb.it}}

\affil[1]{University of Urbino Carlo Bo, Italy}
\affil[2]{TU Wien, Austria}
\date{}

\begin{document}

\maketitle
\begin{abstract}
%In light of the critical role of decision trust, numerous formalisms have emerged. Here, we present a novel logical framework that diverges from existing approaches by accommodating the integration of multiple factors in trust emergence. Our approach relies on 
%a modality for belief and a novel non-monotonic conditional operator that establishes a positive support relation between statements. 
We present \logic, a logical framework designed to formalize decision trust. Our logic integrates a doxastic modality with a novel non-monotonic conditional operator that establishes a positive support relation between statements, and is closely related to a known dyadic deontic modality. For \logic, we provide semantics, proof theory and complexity results, as well as
motivating examples.
%examples of its use. 
Compared to existing approaches, our framework seamlessly accommodates the integration of multiple factors in the emergence of trust.
%\mirko{allowing different potential support relations to occur between statements and by integrating the doxastic modality with such support operator. In the paper we provide all the technical details related to the framework and We show that our approach has clear advantages in modeling decision trust.}
%accommodate the integration of multiple factors in trust emergence.
%$\supp$ turns out to be the (non nested) dyadic obligation of a known preference-based deontic logic. 
 
\end{abstract}

\section{Introduction}\label{sec:intro}

\textit{Decision trust} is defined as the willingness to depend on something (or somebody) with a feeling of relative security, although negative consequences are possible~\cite{JOSANG2007618,McKnight2000THEMO}.
This notion plays a central role in 
computer-mediated interactions.
%
%%%%%%%%%in online decision-making processes, 
% For instance, for e-commerce 
% when there is an abundance of vendors in a marketplace offering nearly identical products, customers use 
% trust %%%%%%%%%%%%%% as a mediating factor 
% to decide whom to buy from~\cite{Soleimani2022-fb}. 
%%%%%%%%%%%% Hence, it is important to understand how trust can emerge in those scenarios and which characteristics a vendor should strive to achieve to gain customers' trust.
%
%In particular, the trustworthiness of a vendor hinges on the combination of various factors, \cite{YOON2015352,WE2005}, including the choice of the vending platform, the vendor's reputation, or the security of payment processes.
%%%%%%%These include the choice of the vending platform, the vendor's reputation, the delivery logistics, and the security of payment processes. % , and the billing procedure.
%
%Analogous considerations emerge in completely different scenarios. For instance, 
%%%%%%when services are exchanged among artificial agents.
%%%%%%%For instance, 
%in the next generation Internet of Things, smart sensors, edge computing nodes, and cloud computing data centers collaborate through services such as data routing and data analytics. In such an ecosystem, interactions are governed by trust estimations that depend on various conditions, e.g., Quality of Service (QoS), see 
%%%%%% or reputation Quality of Service (QoS), security guarantees, real-time processing, and so on, see, 
%\cite{10102745,10.1145/3558779}.
%
For instance, in e-commerce, when there is an abundance of vendors in a marketplace offering nearly identical products, customers use trust to decide whom to buy from~\cite{Soleimani2022-fb}. Similarly,  in the next generation Internet of Things, smart sensors, edge computing nodes, and cloud computing data centers rely on trust to share services such as data routing and analytics~\cite{10.1145/3558779}. In spite of their differences, in both scenarios, interactions are governed by trust evaluations that depend on various conditions, e.g., security-based policies, reputation scores, Quality of Service (QoS), and the trustee's (avail)ability to behave as expected by the trustor.

Those facts drove the development of various formal models for assessing trust, % computations and manipulations, 
see, e.g., \cite{TagliaferriOslo,10.1145/3236008,10.1007/978-3-540-74810-6_8}. Yet, each existing model relies on specific conditions for the emergence of trust.
%, e.g., security-based policies, reputation scores, or the trustee's (avail)ability to behave as expected by the trustor. 
The conditions are specifically selected depending on the Trust model in use and then applied to a given %application 
domain as fundamental requirements enabling trust.
However, this specialized approach fails to work in environments
where many %(if not all) of these 
conditions contribute to the emergence of trust, see, e.g., the Forbes report~\cite{Forbes}.
%. For example, as stated in a Forbes report~\cite{Forbes}, customers in a digital marketplace have a multifactorial approach to their evaluations: many (if not all) the conditions we mentioned earlier contribute to the emergence of trust.
This calls for Trust models that can express multifaceted information combined to evaluate the presence or lack of trust in the environment~\cite{10.1093/logcom/exac016}.
%There are situations, however, in which trust emerges from combinations of factors 
%considered by different models.
To address this need, we introduce \logic, a logic that allows reasoning about decision trust relying on varied enabling conditions.
%\agata{This paper fills in the gap.}
%We propose a framework based on a logical formalism for reasoning about a notion of computational decision trust relying on varied enabling conditions.
%We argue that a more general approach is needed, where various conditions can be combined. 
% and allow high-level reasoning about computational trust.
%\agata{None of the existing models, however, interprets trust as an acceptance attitude and demonstrates its emergence.}
%\mirko{Existing logical models, however, are heavily application-specific and are ill-suited to reason about trust in a general way.}
%In this paper, we introduce a sound and complete logical formalism for high-level reasoning about computational trust. 
In our logic, trusting a formula $\varphi$ means that the trustor is willing to accept the formula as being true, although it might be false. This acceptance-based interpretation of trust % aligns with the definition of decision trust surveyed above.
%Moreover, it 
is compatible with influential conceptual analyses of the notion of trust that show that trusting a proposition boils down to using the proposition as a premise in one's reasoning, even though the proposition might be false~\cite{FrostArnold2014}.

Concretely, in \logic, Trust is a derived operator whose constituents are a Support connective and a belief operator
(hence the logic's name). Whenever it is both believed that a formula $\varphi$ supports a formula $\psi$, and that $\varphi$ is true, then $\psi$ is $\varphi$-trusted.
%
%Hence, our trust operator ($T_\varphi \psi$) is built using those ingredients - $(B(\varphi) \land B(\varphi \supp \psi)) \rightarrow T_\varphi(\psi)$.
The notion of support, establishing a form of positive influence between two statements, is modeled through a novel dyadic %support 
operator $\supp$, where
$\varphi \supp \psi$ is read as: in the most likely $\varphi$-scenarios $\psi$ holds.
%We will, however, shorten this reading to: given $\varphi$ then $\psi$ is most likely. The idea is that when evaluating a formula of the form $\varphi \supp \psi$, only the most likely $\varphi$-scenarios are considered, and, in those, it must be checked whether $\psi$ holds. This approach is inspired by preference-based logics, see \cite{grossi2022reasoning}, in which a conditional statement ``If $\varphi$ then $\psi$'' is interpreted as among the ``best'' possible scenarios in which $\varphi$ is true, $\psi$ is true as well.
The operator $\supp$ yields a non-monotonic conditional sharing properties with the KLM logic $\Plogic$ of preferential reasoning~\cite{KLM}, but cautious monotonicity. Furthermore, it encompasses additional properties.
%that shares the properties of the KLM logic $\Plogic$ of preferential reasoning~\cite{KLM} but cautious monotonicity. 
We characterize $\supp$ with semantics and proof theory; its axioms and rules turn out to axiomatize 
the flat (i.e., non-nested) fragment of \AA qvist system $\F$~\cite{A84} -- a foundational preference-based logic for normative reasoning.
The notion of belief (what is considered to be true from a subjective standpoint) is modeled through the classical belief operator $B$, obtained through the normal modal logic $\KDf$.
%which characterizes Standard Deontic Logic~\cite{von1951deontic}. 
%The notion of support, establishing a form of positive influence between two statements, will be modeled through a novel dyadic support operator $\supp$, such that $\varphi \supp \psi$ expresses that under the condition $\varphi$, it is more likely that $\psi$ holds.
Hence, our Trust operator ($T_\varphi \psi$) is built using those ingredients - $(B(\varphi) \land B(\varphi \supp \psi)) \rightarrow T_\varphi(\psi)$.

%In the following, through a comparison with the literature, we provide motivations for introducing \agata{yet} a new logical formalization of trust (Section~\ref{sect:art}) and a novel support operator (Section~\ref{sect:supp}). 
In the following, through a comparison with the literature, we provide motivations for introducing yet a new logical framework for decision trust. The key ingredient is the support operator $\supp$, for which we discuss in Sect.~\ref{sect:supp} the (undesired and) required properties.
For our logic we present syntax (Sect.~\ref{sect:syn}), semantics (Sect.~\ref{sect:sem}), and establish the connection between $\supp$ and \AA qvist system $\F$. Soundness, completeness, and complexity (for both the satisfiability and the model checking problem) for \logic\ are established in Sect.~\ref{sect:res}.
%Sect.~\ref{sect:con} concludes the paper.
%, %(Section~\ref{sect:dec}), 
 %and highlight some conclusions (Section~\ref{sect:con}).

%\textcolor{red}{Comparisons (here or in a separate section.)}
%\textcolor{red}{Overview of the paper.}dd
%\mirko{I might have found a few recent surveys on computational trust as applied to e-commerce. Should I include those or look for more general surveys?}

\subsection{Decision trust: state of the art}
\label{sect:art}

Logical formalisms for decision trust can be classified into one of the following three paradigms~\cite{ARTZ200758}:

\begin{itemize}
    \item \textbf{Policy-based models}: trust is obtained by implementing hard-security mechanisms based on cryptographic protocols and access control, see, e.g., \cite{PolicyTrust}. Logical frameworks for policy-based mechanisms are defined 
    in, e.g., \cite{Abadi2009,10.1007/978-3-642-28641-4_2}.
    \item \textbf{Reputation-based models}: trust is obtained through indications of past interactions that are evaluated by gathering and manipulating performance scores for those interactions, see, e.g., \cite{9268769}. Logical approaches in this setting are, e.g., \cite{ALDINI2024109167,LIU2011547,Pinyol2012}.
    \item \textbf{Cognitive models}: trust derives from the combination of various complex factors,
    %focus more on how various combinations of factors influence trust, rather than providing actual applicable models that show how to obtain such factors or how to explicitly combine them. For example, 
    including the agent's disposition and the importance/utility of a situation~\cite{Marsh1994FormalisingTA}, or the agent's expectation and willingness~\cite{CastelfranchiFalcone2010}; several logics %intend to 
    formalize such cognitive aspects~\cite{10.1007/978-3-030-39749-4_1,10.1093/jigpal/jzp077,10.5555/3237383.3237415,DBLP:conf/prima/LiuL17,Primiero2020-PRIALO-3}.
\end{itemize}

%Although each class of models has its proper application domains, they all have potential shortcomings.

Although models that fall within one given paradigm 
are employed in real-world applications, %have their value,
see, e.g., \cite{kamvar2003eigentrust}, they tend to rely on partial features of trust or assume extremely specific conditions, thus are limited. %Therefore, as soon as complex scenarios come into play and trust can emerge through different paths, those models become limited.
Policy-based models flatten trust on the use of (cryptographic) protocols and regulations
%hard-security mechanisms (e.g.,  cryptographic protocols) and regulations,
%or stringent documentation requirements), 
that fail whenever they circularly rely on some trust conditions - \textit{the problem of trusting the policy-makers}~\cite{10.1007/978-3-662-43813-8_1}. 
Reputation-based models flatten trust on scores that often represent only a proxy for trust - \textit{the problem of the insufficiency of reputation for trust}~\cite{10.1145/3236008}. 
Differently from the other paradigms, cognitive models of trust can combine various ingredients that reflect agents' cognitive states, thus capturing a more nuanced notion of trust. However, those models rely on specific definitions of trust taken from cognitive science (see, e.g., the logical model of~\cite{10.1093/jigpal/jzp077}, which is inspired by the cognitive theory of trust studied in~\cite{CastelfranchiFalcone2010}) and specify necessary conditions for trust emergence. This creates a trade-off between the effectiveness in modeling various aspects of trust and the complexity in estimating all its constituting elements in real-world  environments.
 The following example better clarifies the difference between the various paradigms.
% we introduce a use
%case that will be employed as a running example throughout the paper.

%Customers might trust vendors with poor reputation (e.g., if the price of a product is low enough to accept the risk of a negative experience) or might also not trust vendors with positive reputation (e.g., if they believe the vendor manipulated the ratings). 

\begin{example}\label{ex:amazon}
As a leading global online retailer, Amazon prioritizes building consumer trust to drive transactions on their platform. To this end, the company enforces various protocols and vendor rules.
Imagine a customer assessing whether to trust the proposition ``Amazon vendor $V_i$ is reliable'' ($\boxed{\mathit{GoodV_i}}$).
%Being one of the largest online retailers in the world, Amazon is particularly interested in developing consumers' trust to increase the number of transactions happening on their website. To achieve its goals, the company implements many protocols and establishes rules that a vendor must abide by to sell on Amazon's platform.
%Consider a customer who has to decide whether to trust the fact that an Amazon vendor $V_i$ is good (modeled by proposition $\mathit{GoodV_i}$).
%if hard-security mechanisms have to be built every time a (small) transaction has to take place, the economical cost of those transactions would skyrocket, making them cost-ineffective.
%In addition, reputation-based models have a huge difficulty in properly dealing with the entry of new agents in the marketplace since agents with no ratings would struggle to find customers.
%
In a {\em policy-based model} of trust, the customer would only be able to trust $\mathit{GoodV_i}$ under the conditions that $V_i$ successfully fulfills Amazon's internal policies (e.g., $V_i$ is a registered company). 
Yet, this approach has various drawbacks: $(i)$ the requirements might be tricked, giving the customer a false sense of security and exposing her to scams; $(ii)$ building trust goes beyond regulations, and customers' trust seldom rely only on vendors abiding to legal and technical policies; $(iii)$ the problem of trust computing would only be shifted from the vendor $V_i$ to Amazon and the policies enforced by it.
%abiding by some of the policies might be impractical, thus limiting the possibility for smaller vendors to enter the marketplace.
%Under a policy-based model of trust, the customer would trust $\mathit{GoodV_i}$ only in case hard-security protocols are in place (e.g., identification procedures, privacy protection, and so on). This may be impractical to achieve or even insufficient, as such protocols do not guarantee, \textit{per se}, the trustworthiness of the vendor and the reliability of a satisfactory service.
%These aspects would not be captured, e.g., by formal models for access control or protocol verification.
%
In a {\em reputation-based model}, trust in $\mathit{GoodV_i}$ can only depend on $V_i$'s positive reviews. However, this has two limitations: $(i)$ new vendors lack reviews, hindering their ability to establish trust; $(ii)$ reviews can be manipulated, leading to inaccurate trust assessments (e.g., in 2017, \textit{The Shed at Dulwich} restaurant became London's number one restaurant on Tripadvisor, although serving fake food).
%the customer would only be allowed to trust $\mathit{GoodV_i}$ under the condition that $V_i$ has good reviews. Again, this approach is limited: $(i)$ new vendors entering the marketplace would have no reviews and therefore could never prove to be good vendors; $(ii)$ reviews can be manipulated, tricking customers into wrong trust estimations. 
%A model purely based on reputation would not be sufficient to consider other factors during the decision-making process, e.g., the propensity of the customer to accept risks whenever information is considered to be scarce.
%
Using a {\em specific cognitive model} of trust, it would be possible to compute trust estimations based on cognitive features of the agents involved (e.g., the intention of the vendor to provide a good service). However, by having to choose a specific cognitive model, the features that can be modelled as trust triggers would be limited to the ones indicated by the model itself. Moreover, cognitive models often neglect to include features typical of the other paradigms, i.e., policies and reputation.
\end{example}

\section{Support operator}\label{sect:supp}

%In this section, 
We introduce the support operator $\supp$. We compare our modeling approach with other approaches used in the literature on non-monotonic reasoning and motivate
the axiomatization we have chosen for our operator.
Henceforth, we will shorten the reading of $\varphi \supp \psi$ to: given $\varphi$, then $\psi$ is most likely.%\dom{In the case of $\top \supp \psi$, we will simply say: $\psi$ is most likely.}

%The intuition is that when evaluating a formula of the form $\varphi \supp \psi$, only the most likely $\varphi$-scenarios are considered, and, in those, it must be checked whether $\psi$ holds. This approach is inspired by preference-based logics, see \cite{shoham_preference_models}, in which a conditional statement ``If $\varphi$ then $\psi$'' is interpreted as among the ``best'' possible scenarios in which $\varphi$ is true, $\psi$ is true as well.

%\dom{In  shorten  reading to: given $\varphi$, then $\psi$ is most likely. }

%\dom{Before we discuss the specific features of our notion of support, we establish a more intuitive reading for ``$\varphi$ supports $\psi$'' which we expand on in the semantics (see Section \ref{sect:sem}). In particular, $\varphi \supp \psi$ should be read as: In the most likely $\varphi$-scenarios $\psi$ holds. We shorten this reading to: Given $\varphi$ then $\psi$ is most likely. The idea is that when evaluating a formula of the form $\varphi \supp \psi$ an agent only considers the most likely $\varphi$-scenarios and, in those, it checks whether $\psi$ holds.}

% we need to define decision trust.

%\mirko{Explanation of what we want/need (more likely to be true) and why this is fundamental for trust.....and why support should be different from relevance.}

\subsection{Why yet another notion of support}

Various notions of support and axiomatizations 
%of the notion of support 
as a conditional operator have been introduced in the literature; see, e.g.~\cite{CrupiIacona2022b} for three potential readings of it
%the conditional connective 
as an evidence operator, or~\cite{3e750f7c-5e82-3f6d-a063-eb4c91160c5f} for a thorough discussion on a dyadic operator for relevance. What all authors agree upon is that the operator should be non-monotonic, i.e., given $\varphi \supp \psi$, there is no reason why $\varphi \wedge \xi \supp \psi$ should be the case. This is because additional information ($\xi$) may undermine 
the previously established supporting statement.
%This is because the supporting relation does not establish the truth of the supported statement but only makes it more likely; this leaves open the possibility of obtaining novel information that undermines the likelihood of the supported statement, thus defeating previously established connections. 
We also assume that support is a non-monotonic operator. 
However, we make assumptions that distinguish our view from the existing ones. 
Specifically, contraposition ($(\phi \supp \psi) \rightarrow (\neg\psi \supp \neg\phi)$) and right weakening (if $\phi\models\psi$, then $\chi\supp\phi\models\chi\supp\psi$) together give monotonicity for the $\supp$ operator.
Differently from~\cite{CrupiIacona2022a}, to avoid monotonicity, we give up contraposition (as motivated through Example~\ref{ExampleCMP}) rather than right weakening, which is a reasonable assumption (see Sect.~\ref{sec:axiomatizationsupport}).
%As shown in~\cite{CrupiIacona2022a}
%However where objective interpretations favor contraposition over right weakening with the $\supp$ operator, our likelihood-based perspective suggests the opposite.
%However, where more objective readings of the $\supp$ operator push towards abandoning right weakening in favor of contraposition, our \textit{likelihood} readings of $\supp$, suggest that contraposition should be abandoned instead. 
%This follows from the fact that the second element of the dyadic operator is not forced upon us by just having the first element. 

\begin{example}[Contraposition and Modus Ponens]
\label{ExampleCMP}
Assume that $\mathit{GoodV_i}$ supports that $V_i$'s products are delivered fast ($\boxed{\mathit{FastV}_i}$), i.e.,
$\mathit{GoodV_i} \supp \mathit{FastV}_i$. This should not imply that if the delivery is slow, then it is most likely that the vendor is not a good one ($\neg \mathit{FastV}_i \supp \neg\mathit{GoodV_i}$), as the delay may depend on other reasons.
%e.g., $V_i$ is dealing with many orders. 
For analogous reasons,
we do not have that $\mathit{GoodV_i}$ and $\mathit{GoodV_i} \supp \mathit{FastV}_i$ imply $\mathit{FastV}_i$, meaning that Modus Ponens for $\supp$ does not hold. 
\end{example}

To proceed methodically, we draw upon the axiomatizations of non-monotonic conditionals introduced in~\cite{KLM}, commonly referred to as KLM systems, as they serve as cornerstones for non-monotonic reasoning.
%we use as a reference the influential axiomatizations for non-monotonic conditionals introduced in~\cite{KLM}, and known as KLM systems.
We start by illustrating with an example why the KLM principle of cautious monotonicity $\mathbf{CM}$ is unsuitable for formalizing our concept of support (see also Remark~\ref{CMrejection}):
$$\mathbf{CM}  \quad ((\varphi \supp \psi) \land (\varphi \supp \chi)) \rightarrow ((\varphi \land \psi) \supp \chi)$$

\begin{example}[{\bf CM}]\label{Example:CM}
Let $\boxed{\mathit{Def}CV_i}$ denote that customer $C$ receives a defective item from vendor $V_i$, and $\boxed{\mathit{Ref}CV_i}$ that $C$ is refunded by $V_i$. In the Amazon marketplace, we have that $\mathit{Def}CV_i \supp \mathit{Ref}CV_i$. 
%Moreover, let $\boxed{\mathit{Buy}CV_i}$ denote that $C$ buys again from $V_i$. 
We also have that given $\mathit{Def}CV_i$ then it is most likely that $V_i$ is not reliable ($\mathit{Def}CV_i \supp \neg\mathit{Good}V_i$). However, having $\mathit{Def}CV_i \land \mathit{Ref}CV_i$ does NOT mean that $\neg\mathit{Good}V_i$ is most likely.
In fact, receiving a refund eases the customer into considering the vendor a good one, and this invalidates {\bf CM}. Example~\ref{Example:semantics} will show the failure of this inference in our logic.
%, the failure of this inference.
%\mirko{Eliminate RM example?}
%new vendor $V_1$ that just entered Amazon's marketplace has an average rating of 4.5 stars for her main product ($\mathit{GoodR_1}$). Having high reviews on a product for a new vendor can support two distinct things: $(i)$ the vendor is indeed a good vendor, $(ii)$ the vendor is inflating her ratings by purchasing reviews ($\mathit{InflatingV_1}$). %Those relationships are captured by the formula ($\mathit{GoodR_1} \supp \mathit{GoodV_1}) \land (\mathit{GoodR_1} \supp \mathit{InflatingV_{1}}$). 
%However, $(\mathit{GoodR_1} \supp \mathit{GoodV_1}) \land (\mathit{GoodR_1} \supp \mathit{InflatingV_1})$ does not imply $(\mathit{GoodR_1} \land \mathit{GoodV_1}) \supp \mathit{InflatingV_1}$, thus making (CM) unsuitable.
%i.e., both $\neg(\mathit{GoodR_1} \land \mathit{GoodV_1} \supp \mathit{Inflating_{V_1}})$ and $\neg(\mathit{GoodR_1} \land \mathit{Inflating_{V_1}} \supp \mathit{GoodV_1})$.
%A similar reasoning can be applied for (RM). Indeed, $\mathit{GoodR_1}$ should not support (by itself) that the vendor is not inflating her ratings (i.e., $\neg(\mathit{GoodR_1} \supp \neg(\mathit{InflatingV_1}))$. Again, assuming that it might be the case that $\mathit{GoodR_1} \supp \mathit{GoodV_{1}}$, we should not be able to derive that $\mathit{GoodR_1} \land \mathit{InflatingV_1} \supp \mathit{GoodV_{1}}$.
\end{example}

%\begin{comment}
%Also the axiom $$(\mathit{CUT}) \quad ((p \supp q) \land (p \land q) \supp r)) \rightarrow (p \supp r)$$ is not desirable in our context. Indeed, even though we assume that high-level QoS standards support using $E_1$ and that these two facts together support uploading data to the cloud-based data center ($\psi_1 \supp \varepsilon_1$ and $(\psi_1 \land \varepsilon_1) \supp \chi_c$), we cannot conclude that high-level QoS standards, by themselves, support uploading to the cloud-based data center. 
%\end{comment}

%Thirdly, support does not subsume implication.
%Even though we assume $\psi_1 \supp \varepsilon_1$ we do not have $\psi_1 \rightarrow \varepsilon_1$, because observing high QoS metrics on the link to $E_1$ does not mean that $E_1$ will be used for the data analytics service.

\subsection{Intuitive properties}  %{\mirko{Axiomatization}}
\label{sec:axiomatizationsupport} 
We introduce the properties
that we envision for the concept of support, illustrating their rationale  refining the scenario outlined in Example~\ref{ex:amazon}. 
As will be shown in Theorem~\ref{teo:Minimal_System}, many of the properties discussed below are inter-derivable, leading to a more concise axiomatization for $\supp$.

Henceforth, by axioms, we mean axiom schemata. 
%We use the derived connective $\varphi\IFF \psi$ as an abbreviation of $(\varphi\to\psi) \wedge (\psi\to\varphi)$, 
%In the following, we use $\models$ to denote the classical consequence relation.
The naming conventions for the considered properties are taken from the KLM systems~\cite{KLM} and $\F$~\cite{A84,Xav2015}. 
 
%During the discussion of our support system, we realised that we axiomatised a flat fragment of an existing Logic called \AA quist's system \textbf{F}. Proof of this will be given in the next section. For that reason, some of the following axioms have the naming conventions taken from \textbf{F}.
 
%\subsubsection*{Axioms.}

As the support operator $\supp$ applies to boolean formulas we expect all classical tautologies to be provable.
Moreover, since any fact intuitively supports itself, the axiomatization of $\supp$ should contain the following axiom: 
$$\mathbf{ID}: \varphi\supp\varphi$$
The presence of this axiom highlights that $\supp$ does not establish a causal relation, see Remark~\ref{rem:connections}.
%since (ID) normally fails for causal interpretations of support.}
%(but neither in the closely related logics I/O logics~\cite{MT00} nor in their causal version~\cite{causal_logic}) 
%
%Let us consider ow the edge computing ecosystem in Example~\ref{ex:runningexample}. If 
%\smallskip
%Consider our edge computing ecosystem, and let us suppose to have another edge node -- $E_3$ -- that does not provide service $s$. This fact motivates that $\varepsilon_3$ supports a contradiction. Therefore, in such a case, we want to derive that $\neg\varepsilon_3$ holds. 
%\begin{example}%[ctd. from Example~\ref{ex:amazon}]
%Let $\boxed{\mathit{Compliant}V_i}$ mean that vendor $V_i$ is compliant with \textit{``Amazon Seller Terms and Conditions''}, and assume that Vendor 1 is indeed compliant, i.e., $\mathit{Compliant}V_1$. From this fact, it follows directly that the fact that vendor 1 is compliant cannot be absurd, expressed in our language through the formula $\neg(\mathit{Compliant}V_1 \supp \bot)$.  
%\end{example}
%This example leads to the factivity axiom:
Moreover we want our support system to not support contradictions.
In essence, anything supporting a contradiction must be dismissed. This principle is reflected in the following axiom:
%
%In other words, whatever supports a contradiction cannot hold. This leads to the axiom
$$\mathbf{ST}: (\varphi \supp \bot) \to \neg\varphi $$
%Suppose that $\psi_1$ and $\phi_1$ support $\varepsilon_1$, because QoS and security guarantees offered by $E_1$ support the choice of $E_1$. Then, it is reasonable to assume that high-level QoS standards for $E_1$ support that offering security guarantees in $E_1$ implies the choice of $E_1$ ($\psi_1 \supp (\phi_1 \rightarrow \varepsilon_1)$). 

\begin{example}
Let $\boxed{\mathit{Compliant}V_i}$ mean that vendor $V_i$ is compliant with \textit{``Amazon Seller Terms and Conditions''} and let $V_i$ be a vendor with an average rating of 4.5 stars for her main product $j$ ($\boxed{\mathit{TopRating}V_{i,j}}$).
Assume that $\mathit{Compliant}V_i$ and $\mathit{TopRating}V_{i,j}$ support that $V_i$ is a good vendor, i.e., $(\mathit{Compliant}V_i\land\mathit{TopRating}V_{i,j})\supp \mathit{Good}V_i$. This implies that being compliant supports the connection between having good reviews and being a good vendor, as expected by Amazon and their customers, i.e., $\mathit{Compliant}V_i \supp (\mathit{TopRating}V_{i,j} \rightarrow \mathit{Good}V_i)$.
\end{example}

This example leads to the following axiom (first introduced as a rule in~\cite{Shoham}):
$$\mathbf{SH}: ((\varphi \land \psi) \supp \chi) \rightarrow (\varphi \supp (\psi \rightarrow \chi))$$
that expresses the fact that deductions performed under strong
assumptions may be useful even if the assumptions are not known facts.
%Assume that using $E_1$ supports the forward of the final result to the data center ($\varepsilon_1 \supp \chi_c$) and also the delivery of the acknowledgment to the IoT device ($\varepsilon_1 \supp \delta$). Hence, $\varepsilon_1$ supports both $\chi_c$ and $\delta$, thus leading to axiom
%$$\mathbf{AND}:
%((\varphi\supp \psi)\land(\varphi\supp \chi))\rightarrow (\varphi\supp(\psi \land\chi))$$
%This axiom requires the side condition $\psi \land \chi \not\models \bot$. For instance, in our example, assume that all communication channels guarantee high performance ($\psi_4$). This fact may support $\varepsilon_1$ and also $\varepsilon_2$. However, it cannot support $\varepsilon_1 \land \varepsilon_2$, which is a contradiction.
%\begin{example}
%Let $\boxed{\mathit{Fair}V_i}$ mean that 
%$V_i$ abides to the \textit{``Acting Fairly''} policy of the ``Amazon's Code of Conduct'', $\boxed{\mathit{Inflating}V_i}$ mean that $V_i$ is inflating her ratings by purchasing reviews, and $\boxed{\mathit{OutsideCom}V_i}$ mean that $V_i$ contacts customers outside the Amazon's Buyer-Seller messaging tool. Intuitively, we have $\mathit{Fair}V_i \supp \neg(\mathit{Inflating}V_i)$ and, similarly, $\mathit{Fair}V_i \supp   \neg(\mathit{OutsideCom}V_i)$. 
%Note that those are supports and not implications as Amazon might not be aware of misbehaving. These statements should obviously imply that $\mathit{Fair}V_i$ supports $\neg(\mathit{Inflating}V_i) \land \neg(\mathit{OutsideCom}V_i)$. 
%\end{example}
%This example leads to axiom:

It is quite natural to assume that if a statement supports two other statements, it supports their conjunction, as expressed by the axiom below:
$$\mathbf{AND}:
((\varphi\supp \psi)\land(\varphi\supp \chi))\rightarrow (\varphi\supp(\psi \land\chi))$$

%However, note that the axiom can never be applied when the two propositions $\psi$ and $\chi$ are contradictory, due to axiom \textbf{ST}. 

%This follows from the fact that axiom \textbf{ST} would make it impossible to satisfy the antecedent of the $\mathbf{AND}$ axiom. 

Note that due to \textbf{ST}, this axiom can never be used to derive that a non-contradictory statement supports a contradiction. We illustrate this with the following example, involving the lottery paradox~\cite{kyburg1961probability}, a stumbling block for default reasoning systems (see~\cite{Poole89}).

%\agata{However, if the conjunction of the two supported statements is contradictory, this does not work anymore. As shown by the following example, this situation is ruled out by axiom \textbf{ST}}.

\begin{example}
  %\dom{The lottery paradox was originally proposed by Kyburg \cite{kyburg1961probability} and demonstrates the issues of reasoning with uncertainty. 
The paradox states that in a fair lottery, it is rational to assume that each individual ticket is likely not to win. By allowing to infer from two statements being likely that their conjunction is also likely, one concludes that two tickets are likely not to win. By iterating this reasoning, we can infer that it is likely that no ticket will win, which contradicts the fact that a winning ticket exists. %Syntactically this scenario cannot occur in our logic because of the axiom \textbf{ST}. 
The paradox does not apply to $\supp$. Indeed, if we assume that every ticket is most likely not to win, $\top \supp \neg T_i$, we can infer $\top \supp \bigwedge \neg T_i$ by \textbf{AND}. %Being $\bigwedge \neg T_i$ a contradiction we derive $\neg \top$ using \textbf{ST}.
Being $\bigwedge \neg T_i$ a contradiction, using \textbf{ST} we could derive $\neg \top$, which is impossible.
%Therefore we cannot assume that for all $i$,
%$\top \supp \neg T_i$.
%\dom{Since $\bigwedge \neg T_i$ is a contradiction, we derive $\neg \top$ using \textbf{ST}. Hence, $\top \supp \bigwedge \neg T_i$ cannot be derived in our logic, without deriving a contradiction. Actually, on the contrary, because of \textbf{ST}, $\neg(\top \supp \bigwedge \neg T_i)$ is derivable, which implies, by the contraposition of \textbf{AND}, that (at least one) ticket $j$ is NOT most likely not to win  $\neg(\top \supp \neg T_j)$. This, intuitively, means that there is always at least one winning ticket, as it should be.}
  
  %Actually, on the contrary, we can derive $\neg (\top \supp \bigwedge \neg T_i)$ using classical contraposition and \textbf{ST}. Using once again classical contraposition and \textbf{AND} we derive that there must be (at least) one ticket for which for which it is not the case to 
  
  %Hence we cannot \agata{derive} that every ticket is likely not to win. 

  %Semantically, this situation is explained by the fact that likeliest is evaluated on the side of the antecedent. This means that  $\top \supp \neg T_i$ does not state that $\neg T_i$ is likely; it rather states that in every likely scenario, the ticket $i$ will not win. Since there exists a winning ticket, it can not be the case that, in every likely scenario, no ticket wins. }
\end{example}

% \begin{example}
%     \dom{The lottery paradox  \cite{kyburg1961probability} states that even though each individual lottery ticket $T_i$ is likely not to win $\top \supp \neg T_i$, their conjunction will win. Hence we can not infer $\top \supp \bigwedge \neg T_i$. Syntactically this scenario can not occur because of the axiom \textbf{ST}. Semantically, this situation is ruled out by always one ticket will win in the likeliest of scenarios even though the agent might not be aware of it.}
% \end{example}
%under the side condition $\psi \land \chi \not\models \bot$. 
%Note how the side condition might come into play in the example above: if we assume that having good reviews might support both that the vendor is good and that she inflated her ratings, this should not imply that having good reviews simultaneously supports the fact that the vendor is good and is also inflating her ratings.

A support operator should also satisfy the $\mathbf{CUT}$ axiom, as illustrated by Example~\ref{CUTexample} below

    $$\mathbf{CUT}: ((\varphi \supp \psi) \land ((\varphi\land\psi)\supp\chi)) \rightarrow (\varphi\supp\chi)$$
    
\begin{example}
\label{CUTexample}
Let $\boxed{\mathit{Auth}V_i}$ stand for  $V_i$ is authenticated on the Amazon marketplace. 
Obviously, 
%Hence $\mathit{Auth}V_i$ and $\mathit{Compliant}V_i$ together support $\mathit{Good}V_i$, i.e.,
an authenticated and compliant vendor is  most likely to be a legitimate business ($\boxed{\mathit{Legit}V_i}$), i.e.,
$(\mathit{Auth}V_i \land \mathit{Compliant}V_i) \supp \mathit{Legit}V_i$. Moreover, due to Amazon's policies, $\mathit{Auth}V_i \supp \mathit{Compliant}V_i$. This implies that given $\mathit{Auth}V_i$ it is already most likely that $V_i$ is legitimate, $\mathit{Auth}V_i \supp \mathit{Legit}V_i$.
\end{example}

\begin{example}
Let $\boxed{\mathit{Fair}V_i}$ mean that 
$V_i$ abides to the \textit{``Acting Fairly''} policy of the ``Amazon's Code of Conduct''.
Consider the case in which we have both 
$\mathit{Compliant}V_i \supp \mathit{Good}V_i$
and $\mathit{Fair}V_i \supp \mathit{Good}V_i$.
These two facts imply that it should be sufficient to satisfy $\mathit{Compliant}V_i$ or  
$\mathit{Fair}V_i$ to be considered a good vendor, i.e., $(\mathit{Compliant}V_i\lor\mathit{Fair}V_i)\supp\mathit{Good}V_i$. 
\end{example}

This example leads to the axiom:
$$\mathbf{OR}:
((\varphi \supp \psi)\land (\chi\supp\psi))\rightarrow ((\varphi \lor\chi)\supp\psi)$$
%Now, suppose that $E_1$ and $E_2$ are indistinguishable because they are two instances of the same product and share the same configurations. Therefore the user is indifferent between $\varepsilon_1$ and $\varepsilon_2$, which ...
%We want to derive that using $E_1$ supports the delivery to the data center ($\varepsilon_1 \supp \chi_c$) if and only if the same holds for $E_2$. 

\begin{example}
Let $\boxed{\mathit{GoodQoS}V_{i,j}}$ mean that vendor $V_i$ offers high QoS while selling product $j$. Consider a situation in which $V_i$ sells two distinct products $a$ and $b$, using the same commercial infrastructure (same logistics, same customer care, and so on). Hence, it would be absurd that $V_i$ offers high QoS only for one of the two products, i.e., $\neg(\mathit{GoodQoS}V_{i,a} \leftrightarrow \mathit{GoodQoS}V_{i,b})\supp\bot$.
Hence whatever $\mathit{GoodQoS}V_{i,a}$ supports, it should also be supported by $\mathit{GoodQoS}V_{i,b}$, and viceversa.
%As a consequence, it is expected that $\mathit{GoodQoS}_{i,a}$ supports that customer $C$ will have a good purchasing experience with $V_i$ ($\boxed{\mathit{GoodExp}CV_{i}}$), if, and only if, the same supporting relation holds between $\mathit{GoodQoS}_{i,b}$ and $\mathit{GoodExp}CV_{i}$, i.e.,
%From this we derive that if offering such a QoS supports that customer $C$ has a good purchasing experience (proposition $\mathit{GoodExp}C_{i}$), that support will not change regardless of which product is sold, whether $a$ or $b$, 
%$(\mathit{GoodQoS_{i,a}}\supp\mathit{GoodExp}CV_{i})\leftrightarrow(\mathit{GoodQoS_{i,b}} \supp \mathit{GoodExp}CV_{i})$. 
\end{example}

This example leads to the following axiom:%\footnote{This axiom is a stronger version of the {\bf LLE} rule below (see Theorem~\ref{teo:Minimal_System}).}:
$$\textbf{LL+}: (\neg (\varphi \IFF \psi) \supp \bot) \rightarrow ((\varphi \supp \chi) \leftrightarrow (\psi \supp \chi))$$

%\subsubsection*{Rules}

%Take $\phi_1'$ to stand for $E_1$ \textit{uses the MQTT-based TLS security layer} and $\phi_1''$ for $E_1$ \textit{uses the Zigbee AES-based security protocol}. From the perspective of the security guarantees associated with such statements, we can say that $\phi_1'$ and $\phi_1''$ are equivalent. Hence, we want to state also that the two supports $\phi_1' \supp \phi_1$ and $\phi_1'' \supp \phi_1$ are equivalent as well. This leads to the Left Logical Equivalence rule

The first rule we consider is motivated by the example:
%a weaker form of axiom $\textbf{LL+}$. 
%
%\begin{example}
%Let $\mathit{Warranty2V_i}$ mean that vendor $V_i$ offers a 2-year warranty on her products, and $\mathit{EuroWarrantyV_i}$ mean that vendor $V_i$ offers a standard European warranty on her products.
%Consider a vendor $V_1$ in Amazon's European market, where the European standard warranty is 2 years long. Hence, the two statements $\mathit{Warranty2V_1}$ and $\mathit{EuroWarrantyV_1}$ are equivalent.
%If offering a 2-year warranty supports the fact that $V_1$ is good, then also offering the standard European warranty should support that $V_1$ is good, and vice-%versa. 
%\end{example}
%
%This example leads to the following derivation rule:
%$$ \mathbf{LLE}:
%\frac{\varphi\leftrightarrow\psi}{(\varphi\supp\chi)\leftrightarrow(\psi\supp\chi)}$$
%Assume now to have $(\psi_1 \land \phi_1) \rightarrow \varepsilon_1'$, where $\varepsilon_1'$ stands for $E_1$ \textit{meets security and performance ISO standards}. Then, if we assume that implementing in the node $E_1$ a Cisco HyperFlex Edge solution ($\gamma$) supports both $\psi_1$ and $\phi_1$ ($\gamma \supp \psi_1$ and $\gamma \supp \phi_1$), then we want to infer that $\gamma$ supports $\varepsilon_1'$, as stated by the rule below reminding the rule of the normal modal logic K

\begin{example}
Let $\boxed{\mathit{GoodPrice}V_i}$ mean that $V_i$ prices well her products and $\boxed{\mathit{AmazonChoice}_{i,j}}$ that the product $j$ sold by $V_i$ is labeled as an ``Amazon's Choice'' product.
By Amazon's policy, $(\mathit{GoodPriceV_i} \land \mathit{Fast}V_i \land \mathit{TopRating}V_{i,j}) \rightarrow \mathit{AmazonChoice}_{i,j}$.
Now, assume 
$\mathit{Good}V_i \supp \mathit{GoodPrice}V_i$ (a good vendor is most likely to price well her products),
$\mathit{Good}V_i \supp \mathit{Fast}V_i$ (a good vendor is  most likely to deliver her products fast), and $\mathit{Good}V_i \supp \mathit{TopRating}V_{i,j}$ (a good vendor is  most likely to have good reviews). Hence, all of these supports together should imply that $\mathit{Good}V_i$ supports $\mathit{AmazonChoice}_{i,j}$. 
\end{example}

The resulting rule is:
$$ \mathbf{RCK}: \frac{(\varphi_1\land\dots\land\varphi_n)\rightarrow \varphi_{n+1}}{((\psi\supp\varphi_1)\land\dots\land(\psi\supp\varphi_n))\rightarrow (\psi\supp\varphi_{n+1})}$$
of which  Right Weakening (\textbf{RW}),  see Remark~\ref{rem:connections},
represents the particular case $n=1$.

\begin{example}
Let $V_i$ be a vendor selling product $j$ and assume that  $\mathit{GoodQoS}V_{i,j} \supp \mathit{Fast}V_i$ and $\neg(\mathit{AmazonChoice}_{i,j} \supp \mathit{Def}CV_i)$. If these conditions hold, then
customer $C$ will have a good purchasing experience. 
%($\boxed{\mathit{GoodExp}CV_i}$). 
Given that
%truth of 
this implication holds, under the same hypothesis, it follows that $C$ having a negative purchasing experience supports a contradiction. 
%Assume that the production chain employed by a vendor supports the idea that the product is of high quality. Also, assume that the supply chain supports the product being on time (thus it does not support the product being delivered late to the customer). If those conditions hold, a customer will have a positive purchasing experience (receiving a good product on time). Given the truth of this scenario, it follows that the customer having a negative purchasing experience is absurd. 
\end{example}

This example leads to the following ``$\Sf$-like" rule:
$$\mathbf{S5_F}: \frac{((\neg)(\varphi_1 \supp \psi_1) \land \dots \land (\neg)(\varphi_n \supp \psi_n))\rightarrow \chi} {((\neg)(\varphi_1 \supp \psi_1) \land \dots \land (\neg)(\varphi_n \supp \psi_n))\rightarrow  (\neg \chi \supp \bot)}$$
where $(\neg) (\varphi_i \supp \psi_i)$ stands for either $(\varphi_i \supp \psi_i)$ or
its negated version $\neg (\varphi_i \supp \psi_i)$.
The rule is named because, when considered alongside other axioms and rules, $\mathbf{S5_F}$ grants the operator $\supp$ all the properties of an $\Sf$-modality 
for the shallow fragment (see Theorem~\ref{teo:flat-fragmentF}). As shown in Sect.~\ref{sect:res}, $\mathbf{S5_F}$ lets the operator $\supp$ behave locally like an absolute operator, playing a crucial role in the completeness proof.
%(Theorem\ref{teo:completness}).

% More specifically, $\mathbf{S5_F}$ enables the replication of $\Sf$ axioms within a language without nested operators.

% The reason why this rule is called $\mathbf{S5_F}$ is that it gives (in the presence of the other axioms and rule) the operator $\supp$ all the properties of an S5-modality. It does so (without containing nested operators) by stating that if $\supp$ formulas and their negation always imply a propositional statement $\chi$, and then the negated statement $\chi$ always supports a contradiction; in other words, it can not be made true.

%Many of those axioms and rules are shared with the KLM logic \textbf{P} of preferential reasoning. It contains
%the rules {\bf LLE}, {\bf RW} and the axioms {\bf ID}, {\bf CUT}, {\bf AND} and {\bf OR}. An important difference between \textbf{P} and the axiomatization of support is the direct interaction of support formulas and propositional formulas in the rule $\mathbf{S5_F}$. The reason for those additions is that we want our logic to be able to talk about the belief system of our agents, which is defined via classical KD4 logic.

%as seen in the axiom \textbf{ST} and

\begin{remark}
\label{rem:connections} 
(Most of) The axioms and rules discussed above are present in well-known systems. 
For instance, the KLM logic $\Plogic$ of preferential reasoning, which interprets the dyadic operator $\varphi~\NonMonot\psi$ as
``$\varphi$ \textit{typically implies} $\psi$'', contains
the rule {\bf RW} (see below) and the axioms {\bf ID}, {\bf CUT}, {\bf AND} and {\bf OR}. %(plus the Right Weakening rule {\bf RW}, that follows from the axioms and rules for $\supp$). 
I/O logics~\cite{MT00} and their causal versions~\cite{causal_logic}, whose dyadic operator is interpreted as a dyadic obligation and a causal relation, respectively, share {\bf RW}, {\bf CUT}, {\bf AND} and {\bf OR} (but notably not {\bf ID}).
Note that KLM and (deontic and causal) I/O logics  also contain the rule 
$\mathbf{LLE}$ below right: 
$$\mathbf{RW}: \frac{\varphi_1\rightarrow\varphi_2}{(\psi\supp\varphi_1)\rightarrow (\psi\supp\varphi_2)}
\qquad \qquad
\mathbf{LLE}:
\frac{\varphi\leftrightarrow\psi}{(\varphi\supp\chi)\leftrightarrow(\psi\supp\chi)}$$
which is weaker than the axiom $\mathbf{LL+}$.
An important difference between these logics and our $\supp$ operator is the direct interaction of support formulas and propositional formulas due to the rule $\mathbf{S5_F}$.
%\agata{In \logic\ we indeed aim to reason, in addition to support, to  an agent belief system.}
%This is because we aim with \logic\ to 
%also reason about an agent belief system, which is defined later via the $\KDf$ system, in addition to support.

%to enable us to talk about agents' belief system, which is defined later via the $\KDf$ system.
\end{remark}

\section{A logical framework for decision trust}\label{sect:syn}

%Trust is defined as the conjunction of two beliefs: a belief of $\varphi\supp\psi$ and the belief of $\varphi$. Intuitively, if you believe that a fact $\varphi$ supports another fact $\psi$ and you believe that $\varphi$ is true, then you trust $\psi$.

%Trust is a complex phenomenon and the term is employed differently in different scenarios -- see, e.g., \cite{Tagliaferri2023}. This means that a unified treatment of the notion of trust seems unreasonable at best. For this reason, our choice is to just focus on a specific interpretation of the notion of trust, i.e., an acceptance attitude, and use this interpretation to build a computational version of the concept. In this section, we will describe the main elements of our interpretation and put it in perspective with others used to describe computational trust; differences and advantages of our proposal with respect to the existing literature will be discussed. In addition, we will also highlight with this notion of trust has advantages to other related concepts employed in computer science to simulate trust, e.g., reputation. We will conclude the section by showing that our interpretation satisfies some general desiderata of trust.
We introduce the logic \logic\ for reasoning about decision trust. \logic\ is obtained by combining the $\supp$ operator with a belief operator $B$. For the former, we
use a (subset of) the discussed axioms and rules and for the latter a $\KDf$ modality\footnote{Alternatively one could  use a $\KDff$ modality, which, however, would include as a side effect negative introspection;
%% Immagino sia un'opinione, che magari non piace ad Halpern. Puoi mettere la reference in maniera piu' sOft? forse "is undesirable for our purposes"??????
see~\cite{Humberstone2016-HUMPAO-2} for a discussion of why it might be undesirable in scenarios similar to the ones we discuss.}

\subsection{Syntax and axiomatization}\label{sec:axiom}
The language $\mathcal{L}$ of \logic\ consists of a countable set of propositional variables (ranging over $p, q, \dots$), the connectives $\wedge$, $\lor$, $\to$, $\leftrightarrow$ and $\neg$ of classical logic, the binary support operator $\supp$, and the unary belief operator $B$. $\mathcal{L}$ is defined by the following two layers grammar ($*\in\{\land, \lor, \to, \leftrightarrow\}$):
\[\begin{array}{rcl}
\varphi & := & \bot \mid p \mid \varphi * \varphi \mid \neg \varphi  \\
\alpha & := & \varphi \mid \varphi\supp\varphi \mid  B(\alpha) \mid \alpha * \alpha \mid \neg \alpha
\end{array}\]
%
%\subsection{Full proof system}
%
%Every rule and axiom contained in the \textbf{minimal system for positive support} plus the following:
%
%Where $*\in\{\land, \lor, \to, \leftrightarrow\}$. 
We use $\varphi, \psi, \chi, \delta$, and $\pi$ only for propositional logic formulas, and $\alpha$ and $\beta$ for general formulas in $\mathcal{L}$. $\ForT$ and $\ForCL$ will denote the set of formulas of \logic\ and of classical propositional logic, respectively.
%AGATA
%We will denote by $\For_T$ the set of  formulas in $\logic$ and $\L_{CL}$ for the set of all formulas in classical propositional logic.
We identify theoremhood in \logic\ with derivability in its Hilbert-style system.

\begin{definition}
\label{def:SBTTrust}    
\logic\ is obtained by extending any axiom system for propositional classical logic: (indicating its axioms by) ($\mathbf{CL}$) and the Modus Ponens rule {\bf MP},  together with:

\begin{description}
\item[For the support operator:]
The axiom schemata
%
%In the following axioms and rules, $\varphi, \psi$, and $\chi$ are propositional formulas.
%(i.e., they do not contain $\supp$ or $B$). \mirko{This notation is confusing with respect to the grammar above. Analogous problem in the Semantics.}
%\agata{The minimal system for $\supp$ consists of the following axioms and rules (cf. Theorem~\ref{teo:Minimal_System}).} 
%(recall that $\varphi, \psi$, and $\chi$ are propositional formulas)
%Then, we have: 
\[
\begin{array}{ll}
(\mathbf{ID}) & \varphi \supp \varphi \\
(\mathbf{ST}) & (\varphi \supp \bot) \rightarrow \neg \varphi \\
(\mathbf{SH}) & ((\psi \land \chi) \supp \varphi) \rightarrow (\psi \supp (\chi \rightarrow \varphi)) \\
(\mathbf{LL+}) & (\neg (\varphi \IFF \psi) \supp \bot) \rightarrow ((\varphi \supp \chi) \leftrightarrow (\psi \supp \chi))
\end{array}
\]
together with the rules $\mathbf{RCK}$ and $\mathbf{S5_F}$ (whose applications must ensure the resulting formulas to be within the language ${\mathcal{L}}$).

\item[For the belief operator:]

The following axiom schemata
\[\begin{array}{ll}
(\mathbf{KB}) & B(\alpha \rightarrow \beta) \rightarrow (B(\alpha) \rightarrow B(\beta)) \\
(\mathbf{DB}) & B(\alpha) \rightarrow \neg B(\neg \alpha) \\
(\mathbf{4B}) & B(\alpha) \rightarrow B (B( \alpha)) \\
\end{array}\]
and the Necessitation rule for $B$ ($\bf NB$).

%$$\mbox{and the Necessitation rule} \quad (\mathbf{N}) \quad \dfrac{\alpha}{B(\alpha)}$$
%
%\[
%\begin{array}{ll}
%\textrm{Necessitation}~(\mathbf{N}) & \dfrac{\alpha}{B(\alpha)} \\[4mm]
%\textrm{Modus Ponens}~(\mathbf{MP}) & \dfrac{\alpha ~~~ \alpha \rightarrow \beta}{\beta}
%\end{array}\]
%
\end{description}
 
\end{definition}

%The logic \logic\ is the minimal set containing all instances of classical tautologies $(\mathbf{CL})$ together with all instances of the following axiom schemata for $B$ and $\supp$:
%\begin{description}
%\item[For the belief operator:]
%\[\begin{array}{ll}
%(\mathbf{KB}) & B(\alpha \rightarrow \beta) \rightarrow (B(\alpha) \rightarrow B(\beta)) \\
%(\mathbf{DB}) & B(\alpha) \rightarrow \neg B(\neg \alpha) \\
%(\mathbf{4B}) & B(\alpha) \rightarrow B (B( \alpha)) \\
%\end{array}\]
%\item[For the support operator:]\[
%\begin{array}{ll}
%(\mathbf{ID}) & \varphi \supp \varphi \\[1mm]
%(\mathbf{ST}) & (\varphi \supp \bot) \rightarrow \neg \varphi \\[1mm]
%(\mathbf{SH}) & ((\psi \land \chi) \supp \varphi) \rightarrow (\psi \supp (\chi \rightarrow \varphi)) \\[1mm]
%(\mathbf{LL+}) & (\neg (\varphi \IFF \psi) \supp \bot) \rightarrow ((\varphi \supp \chi) \leftrightarrow (\psi \supp \chi))
%\end{array}
%\]
%\end{description}
%and is closed under the rules: $\mathbf{RCK}$, $\mathbf{S5_F}$, Modus Ponens $(\mathbf{MP})$, and Necessitation for $B$ $(\mathbf{N})$.
%\end{definition}

Trust arises as a combination of support and belief. 

\begin{definition}\label{def:trust}
    $T_\varphi \psi:= B(\varphi) \land B(\varphi \supp \psi)$.
\end{definition}

% \begin{definition}\label{def:proof}
%   We call a finite sequence of formulas a \textbf{proof} if each element of the sequence is an instance of an axiom, or follows from earlier items in the sequence by applying a rule. A formula $\alpha \in \Lan$  is \textbf{provable} (in symbols: $\vdash \alpha$) if it is the final element in a proof.
% \end{definition}

% \begin{definition}[Derivability]\label{def:derivable}
%     A formula $\alpha \in \mathcal{L}$ is \textbf{derivable} from a set of assumptions $\Phi \subseteq \mathcal{L}$ (in symbols: $\Phi \vdash \alpha$) if there exists a finite sequence of formulas $\alpha_1,...,\alpha_n$ such that $\alpha_n = \alpha$ and for all $i=1,...,n$ one of the following holds:
%     (a) $\vdash \alpha_i$,
%     (b) $\alpha_i \in \Phi$,
%     (c) There exist $j, k \, < i$ such that $\alpha_i$ is the result of an application of one of the rules to $\alpha_j$ and $\alpha_k$.
% \end{definition}

% \begin{fact}
%     $\vdash \alpha$ iff $\emptyset \vdash \alpha$.
% \end{fact}

The notion of derivation is the usual one (some care is required to maintain the restrictions on our language), as well as the 
notion of derivability for a formula $\alpha$ from a set of assumptions $\Phi$, which we denote as $\Phi \vdash \alpha$. %is defined as usual. 
We write  $\vdash \alpha$ iff $\emptyset \vdash \alpha$. Clearly, in \logic, the deduction theorem holds. % $\vdash \alpha$ iff $\emptyset \vdash \alpha$ and
We now prove that all the axioms and rules stated in Sect.~\ref{sec:axiomatizationsupport} are derivable in \logic.

% \begin{definition}[Derivability]\label{def:derivable}
%     We call a formula $\alpha \in \mathcal{L}$ a \textbf{theorem} (in symbols: $\vdash \alpha$) if $\alpha$ is an instance of an axiom or is the result of an application of one of the rules to a theorem. We say that $\alpha$ is \textbf{derivable} from a set of assumptions $\Phi \subseteq \mathcal{L}$ (in symbols: $\Phi \vdash \alpha$) if there exists a finite sequence of formulas $\alpha_1,...,\alpha_n$ such that $\alpha_n = \alpha$ and for all $i=1,...,n$ one of the following holds:
%     (a) $\alpha_i$ is a theorem
%     (b) $\alpha_i \in \Phi$,
%     (c) There exist $j, k \, < i$ such that $\alpha_i$ is the result of an application of the rule of modus ponens to $\alpha_j$ and $\alpha_k$.
% \end{definition}

%\begin{fact}\label{fact:deduction}

\begin{theorem}
\label{teo:Minimal_System}
The rules {\bf RW} and {\bf LLE}, as well as the axioms {\bf AND}, {\bf CUT}, and {\bf OR} are derivable in the system for $\supp$.
%Every rule and axiom of the positive support system in subsection positive support can be derived in this minimal system.
\end{theorem}

\begin{proof}
Trivial for \textbf{RW}, \textbf{LLE} and \textbf{AND}. The \textbf{CUT} axiom follows by
$\mathbf{RCK}$ applied to formulas obtained by
$\textbf{SH}$  and $\textbf{CL}$.
Axiom \textbf{OR}: 
by applying \textbf{LLE} to
$(\varphi \land (\varphi \lor \chi)) \IFF \varphi$ we get $(\varphi \supp \psi) \IFF ((\varphi \land (\varphi \lor \chi)) \supp \psi$); similarly, we get also
$(\chi \supp \psi) \IFF ((\chi \land (\chi \lor \varphi)) \supp \psi$). The claim follows by two applications of $\textbf{SH}$, followed by
$\textbf{CL}$, $\textbf{RCK}$, $\textbf{ID}$, $\textbf{AND}$ and $\textbf{RW}$.
Full proof in appendix.
\end{proof}

\begin{remark}
\label{CMrejection}
Another reason for rejecting \textbf{CM} is that, in conjunction with \textbf{CUT} and \textbf{RW}, it permits to derive \textbf{REC}:
$((\varphi \supp \psi)\land (\psi \supp\varphi))\rightarrow ((\varphi\supp\chi) \leftrightarrow  (\psi \supp\chi) )$ (derivation can be found in the appendix). \textbf{REC} is too strong for a support operator since two statements that support each other do not necessarily support the same statements. For instance, the proposition $\mathit{GoodV_i}$ and the statement "$V_i$ always responds quickly to a customer's question" may support each other but do not support the same statements; the latter may support that $V_i$ uses an AI tool to answer while the former does not.

%Not receiving a response from vendor $V_i$ makes it most likely that $V_i$ is not reliable ($\neg \mathit{Res}V_i \supp \neg\mathit{Good}V_i$). Furthermore, $V_i$ not being reliable makes it most likely that we will not receive a response ($\neg\mathit{Good}V_i \supp \neg \mathit{Res}V_i$). If we do not receive a response from $V_i$, it is most likely that we will contact Amazon directly ($\neg \mathit{Res}V_i \supp \mathit{Cont}A$). This does not imply that if the $V_i$ is not reliable, it is most likely that we will contact Amazon since there might be a most likely  scenario (WORK IN PROGRESS)
\end{remark}

A strong connection holds between $\supp$ and $\F$, the dyadic deontic logic introduced in~\cite{A84} and axiomatized in~\cite{Xav2015} as follows ($\bigcirc(\psi / \varphi)$ stands for ``$\psi$ \textit{is obligatory under the condition} $\varphi$''):

%We have mentioned at the beginning of this section that the introduced axioms correspond to a flat fragment of the system \textbf{F}. As a first step, we show that every "definable" axiom and rule of \textbf{F} can be derived given the above-defined axioms and rules. Here, the term definable refers to every axiom/rule which does not rely on the nesting of the modal operator. 

% \noindent
% \textbf{Axioms}
% \begin{align*}
% %&\textbf{Axioms:} \\
% &\textbf{CL} ~~ \text{All truth-functional tautologies}    \\
% &\textbf{S5} ~~\text{S5-schemata for $\Box$ and $\lozenge$} \\
% &\textbf{COK} ~~\bigcirc(\psi \rightarrow \chi / \varphi) \rightarrow (\bigcirc(\psi / \varphi) \rightarrow \bigcirc(\chi / \varphi))  \\
% &\textbf{Abs} ~~\bigcirc(\varphi  / \psi) \rightarrow \Box \bigcirc(\varphi / \psi) \\
% &\textbf{Nec} ~~\Box \varphi \rightarrow \bigcirc (\varphi / \psi) \\
% &\textbf{Ext} ~~\Box (\varphi \leftrightarrow \psi) \rightarrow (\bigcirc(\chi / \varphi) \leftrightarrow \bigcirc(\chi / \psi) ) \\
% &\textbf{ID} ~~\bigcirc(\varphi / \varphi) \\
% & \textbf{SH} ~~ \bigcirc(\varphi / \psi \land \chi) \rightarrow \bigcirc (\chi \rightarrow \varphi / \psi) \\
% &\textbf{D*} ~~\lozenge \psi \rightarrow (\bigcirc (\varphi / \psi) \rightarrow \neg \bigcirc( \neg \varphi / \psi)) 
% \end{align*}
% {\bf Rules:} modus ponens \textbf{MP} and necessitation for $\Box$.

\begin{itemize}

\item Axioms:
\[
\begin{array}{rl}
     \mathbf{CL}  & \textrm{All truth-functional tautologies} \\ 
     \mathbf{K_\Box}   & \Box (\varphi \rightarrow \psi) \rightarrow \Box \varphi \rightarrow \Box \psi 
     \\
     \mathbf{T}   & \Box \varphi \rightarrow \varphi
     \\ 
     \mathbf{5}   & \lozenge \varphi \rightarrow \Box \lozenge \varphi 
     \\
     \mathbf{COK} & \bigcirc (\psi \rightarrow \chi / \varphi) \rightarrow (\bigcirc (\psi / \varphi) \rightarrow \bigcirc (\chi / \varphi)) 
     \\ 
     \mathbf{Abs} & \bigcirc (\varphi  / \psi) \rightarrow \Box \! \bigcirc \!(\varphi / \psi) 
     \\
     \mathbf{Nec} & \Box \varphi \rightarrow \bigcirc (\varphi / \psi) 
     \\ 
     \mathbf{Ext} & \Box (\varphi \leftrightarrow \psi) \rightarrow (\bigcirc(\chi / \varphi) \leftrightarrow \bigcirc(\chi / \psi) ) 
     \\
     \mathbf{ID}  & \bigcirc(\varphi / \varphi)
     \\ 
     \mathbf{SH}  & \bigcirc(\varphi / \psi \land \chi) \rightarrow \bigcirc (\chi \rightarrow \varphi / \psi) \\
     \mathbf{D^*}  & \lozenge \psi \rightarrow (\bigcirc (\varphi / \psi) \rightarrow \neg \!\bigcirc \! ( \neg \varphi / \psi)) 
\end{array}
\]

\begin{comment}
\begin{align*}
%&\textbf{Axioms:} \\
&\textbf{CL} ~~ \text{All truth-functional tautologies}    \\
&\textbf{K} ~~ \Box (\varphi \rightarrow \psi) \rightarrow \Box \varphi \rightarrow \Box \psi \\
&\textbf{T} ~~\Box \varphi \rightarrow \varphi \\
&\textbf{5} ~~\lozenge \varphi \rightarrow \Box \lozenge \varphi \\
&\textbf{COK} ~~\bigcirc(\psi \rightarrow \chi / \varphi) \rightarrow (\bigcirc(\psi / \varphi) \rightarrow \bigcirc(\chi / \varphi))  \\
&\textbf{Abs} ~~\bigcirc(\varphi  / \psi) \rightarrow \Box \bigcirc(\varphi / \psi) \\
&\textbf{Nec} ~~\Box \varphi \rightarrow \bigcirc (\varphi / \psi) \\
&\textbf{Ext} ~~\Box (\varphi \leftrightarrow \psi) \rightarrow (\bigcirc(\chi / \varphi) \leftrightarrow \bigcirc(\chi / \psi) ) \\
&\textbf{ID} ~~\bigcirc(\varphi / \varphi) \\
& \textbf{SH} ~~ \bigcirc(\varphi / \psi \land \chi) \rightarrow \bigcirc (\chi \rightarrow \varphi / \psi) \\
&\textbf{D*} ~~\lozenge \psi \rightarrow (\bigcirc (\varphi / \psi) \rightarrow \neg \bigcirc( \neg \varphi / \psi)) 
\end{align*}
\end{comment}

\item Rules: modus ponens \textbf{MP} and necessitation for $\Box$.

\end{itemize}

%\\
%&\textbf{Rules:} \\
%&\textbf{MP} ~~\text{If} \vdash \varphi ~~ \text{and} ~ \vdash \varphi \rightarrow \chi ~~ \text{then} ~ \vdash \chi \\
%&\textbf{N} ~~\text{If} \vdash \varphi ~~ \text{then} ~~ \vdash \Box \varphi 
%\end{align*}

%We use the following translations:
%
%\begin{itemize}
%    \item $\bigcirc(\psi / \varphi) \mapsto  \varphi \supp \psi$
%    \item $\Box \varphi \mapsto \neg \varphi \supp \bot$
%    \item $\lozenge \varphi \mapsto \neg (\varphi \supp \bot)$
%\end{itemize}

%This translation is not problematic since in \textbf{F} the equivalences $\bigcirc(\bot / \neg \varphi) \leftrightarrow \Box \varphi$ and $\lozenge \varphi \leftrightarrow \neg \Box \neg \varphi$ are derivable.

The flat fragment (i.e., $\Box$ and $\bigcirc$ apply only to formulas of $\ForCL$) of the language of $\F$ can be translated into our language as follows:

\begin{definition}\label{def:translation}
Let $\chi$ be any formula in the flat fragment of $\F$. %that does not contain nested modal operators. 
The translation $\chi^*$ using $\supp$ is 
($\varphi, \psi \in \ForCL$):
\[\begin{array}{rclp{1cm}rcl}
\varphi^* & \mapsto & \varphi & &
     (\Box \varphi)^* & \mapsto & \neg \varphi \supp \bot \\
     (\lozenge \varphi)^* & \mapsto & \neg (\varphi \supp \bot) & &
     (\bigcirc(\psi / \varphi))^* & \mapsto & \varphi \supp \psi  
\end{array}\]
%Given a formula $\chi$ in the language of \textbf{F} that does not contain nested modal (\agata{both $\square$ and $\bigcirc$}) operators, we will write from now on $\chi^*$ for the translation into our support language.
\end{definition}

% Observe that even though \textbf{F} contains the modal operators $\bigcirc$, $\Box$ and $\lozenge$ we translate all three into the same operator $\supp$. This does not cause any unwanted derivations due to the equivalences that hold in \textbf{F}, between $\bigcirc(\bot / \neg \varphi)$ and $\Box \varphi$, and between $\lozenge \varphi$ and $\neg \Box \neg \varphi$.

%     Consider the following syntactic translations:
% \[\begin{array}{rcl}
%      \bigcirc(\psi / \varphi) & \mapsto & \varphi \supp \psi \\
%      \Box \varphi & \mapsto & \neg \varphi \supp \bot \\
%      \lozenge \varphi & \mapsto & \neg (\varphi \supp \bot)      
% \end{array}\]

We show that with the exception of $\mathbf{5}$ and $\mathbf{Abs}$, all axioms and rules of $\F$ (within the flat fragment) are derivable
in the axiomatization for $\supp$.
This establishes a first link between $\F$ and $\supp$.

\begin{theorem} 
\label{teo:connectionF}
 The translation $^\ast$ of all axioms and rules of $\F$ -- but $\mathbf{5}$ and $\mathbf{Abs}$ -- are derivable in the axiomatization for $\supp$.
\end{theorem}

\begin{proof}
The claim for axioms $\mathbf{T}$, $\mathbf{Ext}$, $\mathbf{ID}$, and $\mathbf{SH}$  directly follows from the translation. For the remaining axioms: The translation $^\ast$ of \textbf{D}$^\ast$ is 
$\neg (\varphi \supp \bot) \to
\neg ((\varphi \supp \psi) \land (\varphi \supp \neg \psi))$. Its
contraposition $((\varphi \supp \psi) \land (\varphi \supp \neg \psi)) \rightarrow (\varphi \supp \bot) $ is an instance of \textbf{AND}.
The translation of $\mathbf{K_\Box}$ can be derived by two applications of
$\mathbf{ST}$, the axioms 
$\mathbf{CL}$ and the rule
$\mathbf{S5_F}$.
$\textbf{COK}$ follows by      $\mathbf{AND}$, $\mathbf{RW}$ and
$\mathbf{CL}$.
$\textbf{Nec}$ follows 
by $\mathbf{ST} + \mathbf{CL}$, 
together with the rules
$\mathbf{S5_F}$ and 
$\mathbf{SH}
+ \mathbf{CL}$.
The Necessitation rule  for $\Box$ follows by modus ponens using $\mathbf{RCK}$ and 
$\mathbf{ID}$.
 Full proofs in appendix.
\end{proof}

\begin{remark}
\label{rem:F}
The translation of axioms {\bf Abs} and $\mathbf{5}$ from $\F$ results in formulas containing nested applications of the support operator. In Sect.~\ref{sect:sem}, we will see that the axioms and rules for $\supp$ axiomatize the shallow fragment of $\F$. In this regard, the rule $\mathbf{S5_F}$, which does not correspond to any rule known in the literature, does not follow from the remaining axioms and rules for $\supp$, and it is needed to derive (some) flat formulas which hold in $\F$.

%is derivable in this fragment but does not follow from the remaining axioms and rules for $\supp$.
\end{remark}

\begin{example}\label{Example:derivation}
Let us see \logic\ at work.
Assume that a customer $C$ believes the two formulas supported by $\mathit{Def}CV_i$ discussed in Example~\ref{Example:CM}.
If $C$ receives a defective item from vendor $V_1$ ($\mathit{Def}CV_1$), from $B(\mathit{Def}CV_1)$ and $B(\mathit{Def}CV_1 \supp \neg\mathit{Good}V_1)$ we derive 
$T_{\mathit{Def}CV_1} (\neg\mathit{Good}V_1)$. Similarly, it also holds that $T_{\mathit{Def}CV_1} (\mathit{Ref}CV_1)$. 
Now, assume that $C$ does indeed receive the refund, thus $B(\mathit{Def}CV_1 \land \mathit{Ref}CV_1)$. 

%By using \logic\ 
We can show that it is not the case that $C$ trusts $V_1$, $ \neg T_{\mathit{Def}CV_1 \land \mathit{Ref}CV_1} (\mathit{Good}V_1)$. We use the following abbreviations to write a concise derivation: Let $d := \mathit{Def}CV_1$, $r := \mathit{Ref}CV_1$, and $g := \mathit{Good}V_1$, and, by hypothesis, $T_{d} (\neg g) \land T_{d} (r)$, i.e., $B(d) \land B(d \supp \neg g) \land B(d \supp r)$. Hence:
\[\begin{array}{llr}
(1) & (d \land (d \supp \neg g)) \rightarrow \neg(d \supp g) & (\mathbf{ST}\!+\!\mathbf{D^*}) \\
(2) & ((d \supp r) \land \neg (d \supp g))  \rightarrow  \neg ((d \land r) \supp g) & (\mathbf{CUT}\!+\!\mathbf{CL}) \\
%(4) & (d \supp r) \land (d \land r) \supp v)  \rightarrow d \supp v & (\mathbf{CUT}) \\
(3) & (d \land  (d \supp \neg g) \land (d \supp r)) \rightarrow  \neg ((d \land r) \supp g) & (1 \land 2) \\
%(4) &  B(d) \land B(d \supp r) \land B(\neg (d \supp v))  \rightarrow  B(\neg (d \land r) \supp v) & (\textit{N for B}) \\
\end{array}\]
% \[\begin{array}{llr}
% (3) & (B(d) \land B(d \supp \neg v)) \rightarrow B(\neg(d \supp v))& (\mathbf{ST} + \mathbf{D^*}) \\
% %(4) & (d \supp r) \land (d \land r) \supp v)  \rightarrow d \supp v & (\mathbf{CUT}) \\
% (5) & (d \supp r) \land \neg (d \supp v)  \rightarrow  \neg (d \land r) \supp v & (\mathbf{CUT} + \mathbf{CL}) \\
% (6) & B(d \supp r) \land B(\neg (d \supp v))  \rightarrow  B(\neg (d \land r) \supp v) & (\textit{N for B}) \\
% (1) &  B(d) \land B(d \supp \neg v) \land B(d \supp r) \rightarrow & (\textit{Def. of T}) \\
% \end{array}\]
Then, applying rule $\bf NB$ to (3) and using the hypothesis together with axiom $\bf KB$, we derive $B(\neg ((d \land r) \supp g))$. This formula, by axiom $\bf DB$, finally implies $\neg T_{d \land r} (g)$.
Note that the lack of axiom CM impedes to derive $T_{\mathit{Def}CV_1 \land \mathit{Ref}CV_1} (\neg\mathit{Good}V_1)$, as shown in Example~\ref{Example:semantics}.
\end{example}

\section{Semantics}\label{sect:sem} 

%The intuition is that when evaluating a formula of the form $\varphi \supp \psi$, only the most likely $\varphi$-scenarios are considered, and, in those, it must be checked whether $\psi$ holds. This approach is inspired by preference-based logics, see \cite{shoham_preference_models}, in which a conditional statement ``If $\varphi$ then $\psi$'' is interpreted as among the ``best'' possible scenarios in which $\varphi$ is true, $\psi$ is true as well.
%\dom{The intuition for evaluating a formula of the form $\varphi \supp \psi$ is that only the most likely $\varphi$-scenarios are considered, and, in those, it must be checked whether $\psi$ holds. This approach is inspired by preference-based logics, see \cite{grossi2022reasoning}, in which a conditional statement ``If $\varphi$ then $\psi$'' is interpreted as among the ``best'' possible scenarios in which $\varphi$ is true, $\psi$ is true as well.} 
For evaluating a formula of the form $\varphi \supp \psi$, we intuitively consider only the most likely $\varphi$-scenarios and check whether $\psi$ holds in those scenarios. This approach is inspired by preference-based logics~\cite{grossi2022reasoning}, in which a conditional statement ``If $\varphi$ then $\psi$'' is interpreted as among the ``best'' possible scenarios in which $\varphi$ is true, $\psi$ is true as well. Hence the semantics for \logic\ is built on preference-based models \cite{shoham_preference_models} (for the support statements) and standard relational models (for the belief operator).

Also used in KLM logics and in \AA qvist system $\F$, preference-based models are triples $\langle S, \succeq, V \rangle$, where $S$ denotes a set of states, $V$ a valuation function, and the preference relation $\succeq \subseteq S \times S$ orders the states in $S$. In our context, $w \succeq v$ means that the state $w$ is at least as likely as $v$ ($\F$ uses a \textit{better than} interpretation, while the KLM logic $\Plogic$ uses a \textit{preferred to} interpretation). Although similar to $\F$ and $\Plogic$, our approach differs from those for two important aspects. First, in our semantics, we include a component for belief formulas. Second, instead of having a unique preference frame, in a \textit{Trust frame}, the set of states $S$ is partitioned into multiple preference frames $\langle S_i, \succeq_i \rangle$; this allows us to consider different support systems within the same model. Note that, without partitioning, the unique given support system would necessarily have to be believed. 
In our interpretation, a statement $\varphi \supp \psi$ is true in a state $w \in S_i$ if among all the $\varphi$-states in $S_i$ ($||\varphi||_i$) the most likely $\varphi$-states in $S_i$ ($\mathit{most}(||\varphi||_i)$) satisfy $\psi$; $\mathit{most}(||\varphi||_i)$
%an acceptable $\varphi$-world in $S_i$ 
%$\varphi$-world in $S_i$ as a most acceptable $\varphi$-world if it 
consists of the maximal elements in the set of all $\varphi$-states of $S_i$ according to $\succeq_i$.

In line with systems $\F$ and $\Plogic$, our semantics include a limitedness condition: $||\varphi||_i\not=\emptyset \Rightarrow \mathit{most}(||\varphi||_i)\not=\emptyset$.
Limitedness allows us to express when a formula $\varphi$ is impossible in a partition $S_i$ using $\varphi \supp \bot$. Hence, limitedness restricts \logic\ to support only 
non-contradictory options (with the exception of $\bot \supp \bot$). Consequently, an agent will never place trust in a blatant contradiction.\footnote{This is different from trusting contradicting statements separately, which is still possible in \logic.}
The definitions in this section characterize the frames and models on which our logic is based on.

\begin{definition}[Trust frame]\label{def:frame}
   Let $\mathcal{F}:=\langle S, (S_i)_{i \in I}, (\succeq_i)_{i \in I}, R \rangle$ where
   \begin{itemize}
        \item $\langle S, R \rangle$  is a serial and transitive Kripke frame;
       \item  $(S_i)_{i \in I}$ is a partition\footnote{$\bigcup_{i \in I} S_i = S$ and $\forall i, j \in I: ~ S_i \cap S_j = \emptyset$.} of $S$;
       \item For each $i \in I$:  $\langle S_i , (\succeq_i) \rangle$ is a preference frame (therefore $\succeq_i \subseteq S_i \times S_i$).
   \end{itemize}
\end{definition}

%Belief only acts on all of $\langle S, R \rangle$, support acts locally on each $\langle S_i, \succeq_i \rangle$ as it does in \textbf{F}.

%\todo{Add specifications for the particular elements of our definitions - e.g., \textit{com}.}

%As a next step, we will define the \textit{Trust models}. 

%SMOOTHNESS (KLM!!!) \dom{I mentioned it now at the beginning; should I also mention it here?}

\begin{definition}[Trust model and truth conditions]\label{def:model}
   Let $\mathcal{M}:=\langle S, (S_i)_{i \in I}, (\succeq_i)_{i \in I}, R, V \rangle$ where
   \begin{itemize}
       \item $\mathcal{F}:=\langle S, (S_i)_{i \in I}, (\succeq_i)_{i \in I}, R\rangle$ is a Trust frame;
       \item $V: Prop \mapsto 2^S $ is a Valuation function;
       \item For each $i \in I$,  $\langle S_i , (\succeq_i), V \rangle$ fulfills the limitedness condition:
       for every propositional formula $\varphi$
 $$||\varphi||_i\not=\emptyset \Rightarrow \mathit{most}(||\varphi||_i)\not=\emptyset$$
%   \end{itemize}
%\begin{itemize}
 \item $\M, s \models p $ iff $s \in V(p)$;
 \item $\M, s \models \neg \alpha$ iff $\M, s \not\models \alpha$;
 \item $\M, s \models \alpha \land \beta $ iff $\M, s \models \alpha $ and $\M, s \models \beta$;
 \item $\M, s \models \varphi \supp \psi $ iff $\mathit{most}(||\varphi||_i) \subseteq ||\psi||_i$ for $s \in S_i$;
 \item $\M, s \models B(\varphi) $ iff $\forall v: (sRv \rightarrow \M, v \models \varphi) $
\end{itemize}
\noindent where $||\varphi||_i:= \{v \in S_i: \M, v \models \varphi \}$ and \\
$\mathit{most}(||\varphi||_i):= \{s \in ||\varphi||_i: \forall v [(v \in ||\varphi||_i \land v \succeq_i s) \rightarrow s \succeq_i v] \}$.

%We say that $\langle S_i , (\succeq_i), V \rangle$ fulfills the limitedness condition if for every formula $\varphi$ the following holds:
% 
% $$||\varphi||_i\not=\emptyset \Rightarrow most(||\varphi||_i)\not=\emptyset$$
   
\end{definition}

%The three main differences between our semantics compared to reasoning with general preference relations \cite{grossi2022reasoning} are the following. Firstly, we only consider a flat fragment since we do not allow the nesting of the support operator. Secondly, we allow for multiple preference relations. This addition is crucial since it avoids $\supp$ being an absolute operator, allowing our agent to believe in various support systems. Finally, we introduce the additional relation $R$, representing the belief relation among possible worlds.

\begin{remark}\label{rem:Models}
   Three observations: $(i)$ The semantics for formulas containing $\lor, \rightarrow$ and $\leftrightarrow$ are defined through $\neg$ and $\land$ as usual.
    $(ii)$ We do not assume any property on the relations $\succeq_i$ to keep the model as general as possible. $(iii)$ %$\varphi \supp \psi$ should not be read as conditional belief per se. Rather, if $\M, w \models \varphi \supp \psi$ holds, then $w$ is part of a set of states in which $\varphi$ supports $\psi$ being true. An agent may or may not believe $\varphi \supp \psi$, independently from the fact that $\varphi \supp \psi$ holds or not. This gives us the possibility to capture an agent's misinformation.}
   $\M, w \models \varphi \supp \psi$ holds if $w$ is part of a set of states in which $\varphi$ supports $\psi$ is true. An agent may or may not believe $\varphi \supp \psi$, independently from the fact that $\varphi \supp \psi$ holds or not. This gives us the possibility to capture an agent's misinformation.
\end{remark}

The two notions of a formula $\alpha$ being a semantical consequence of a set of formulas $\Phi$ (in symbols: $\Phi \models \alpha$) and $\alpha$ being valid (in symbols: $\models \alpha$) are defined as usual.

%\begin{definition}[Validity]\label{def:validity}
%A formula $\alpha$ is a semantical consequence of a set of formulas $\Gamma$ (in symbols: $\Gamma \models \alpha$) iff for every Trust model $\M$ and every state $s$ we have

% $$ (\forall \beta \in \Gamma ~ \M, s \models \beta)  \Rightarrow \M, s \models \alpha$$

%$\alpha$ is valid (in symbols: $\models \alpha$) iff it is a semantical consequence of the empty set.
 
%\end{definition}

% \begin{definition}[Validity]

%     We call a formula $\alpha$ valid in a model $\M$ (in symbols: $\M \models \alpha$) iff for all $s\in S$: $\M, s \models \alpha$ \\

%     We call a formula $\alpha$ valid in a class of models $\mathbb{M}$ (in symbols: $\mathbb{M} \models \alpha$) iff for all models $\M \in \mathbb{M}$: $\M \models \alpha$ \\

%  We call a formula $\alpha$ valid (in symbols: $\models \alpha$) iff $\alpha$ is valid in the class of all Trust models. \\
 
%     %We call a formula $\varphi$ valid in a frame $\mathcal{F}$ (in symbols: $\mathcal{F} \models \varphi$) iff for all models $\M$ based in the frame $\mathcal{F}$: $\M \models \varphi$ \\

%  We call a formula $\alpha$ a semantical consequence of a set of formulas $\Gamma$ (in symbols: $\Gamma \models \alpha$) iff for 
%  every Trust models $\M$ and every state $s$ we have

%  $$ (\forall \beta \in \Gamma \M, s \models \beta)  \Rightarrow \M, s \models \alpha$$
    
% \end{definition}

% \mirko{Add the example discussed with Dominik. Highlight all the major characteristics of our approach.}

% \dom{Maybe an example not about CM is easier to understand}

\begin{example}[Ctd Ex.~\ref{Example:derivation}]\label{Example:semantics}

%We should change buy again into being a good vendor, then the example becomes much stronger.
% 
%Continuing Example~\ref{Example:derivation}, we 
%
We use our semantics to show that 
$T_{\mathit{Def}CV_1 \land \mathit{Ref}CV_1} (\neg \mathit{Good}V_1)$ (i.e., the customer  trusts that $V_1$ is not a good vendor on the basis of having received a defective item and a refund) does not follow from
 $T_{\mathit{Def}CV_1} (\mathit{Ref}CV_1)$ and $ T_{\mathit{Def}CV_1} (\neg\mathit{Good}V_1)$. 
 %This provides an additional argument for rejecting ({\bf CM}).}
%$\neg T_{\mathit{Def}CV_1 \land \mathit{Ref}CV_1} (\neg \mathit{Buy}CV_1)$ is satisfiable. 
Consider the Trust model $\mathcal{M}:=\langle S, (S_i)_{i \in I}, (\succeq_i)_{i \in I}, R, V \rangle$, with $I=\{1\}$, $S_1:=\{s, w, v\}$, $s \succeq_1 w$ and $w \succeq_1 v$ for support, $v R s, w R w, s R s$ for belief, and $V(\mathit{Def}CV_1):=S$, $V(\mathit{Ref}CV_1):=\{s, v\}$, $V(\mathit{Good}V_1):=\{v\}$. A graphical representation of $\mathcal{M}$ is below; $d := \mathit{Def}CV_1$, $r := \mathit{Ref}CV_1$, $g := \mathit{Good}V_1$, and solid and dashed arrows represent the preference relation $\succeq_1$ and the accessibility relation $R$, respectively.

% The logic for support states that in the case that $d \supp \not b$ and $d \supp \not r$ we can no longer achieve $d \land r \supp b$. This is a desired result since if 

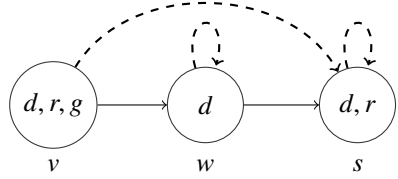
\begin{center}

	\begin{tikzpicture}[
			world/.style={circle, draw=black!80, minimum size=10mm},
			]
			\node[world]      (s)       at (4,0)     {$d,r$};
   \node      (s1)       at (4,-0.8)     {$s$};
            \node[world]      (w)       at (2, 0)     {$d$};
            \node      (w1)       at (2,-0.8)     {$w$};
			\node[world]      (v)       at (0,0)     {$d,r,g$};
   \node      (v1)       at  (0,-0.8)    {$v$};

			\path[->]
			(v)   edge (w)
            (w)   edge (s);

   \path[->,thick,black,dashed]
    (v)   edge [bend left=60] (s)
     (w)   edge  [loop above] (w)
             (s)   edge  [loop above] (s) ;
		\end{tikzpicture}

\end{center}
Let $v$ be the current state. In $v$ the customer $C$ believes $\mathit{Def}CV_1$ and $\mathit{Ref}CV_1$ (i.e. $B(\mathit{Def}CV_1 \land \mathit{Ref}CV_1)$). Since $\mathit{most}(||\mathit{Def}CV_1||_1)=\{s\} \subseteq \{w, s\}=||\neg \mathit{Good}V_1||_1$, we have that $\M, s \models \mathit{Def}CV_1 \supp \neg \mathit{Good}V_1$. Hence, $\M, v \models B(\mathit{Def}CV_1 \supp \neg \mathit{Good}V_1)$ and $\M, v \models T_{\mathit{Def}CV_1} (\neg \mathit{Good}V_1)$.
Similarly, $\mathit{most}(||\mathit{Def}CV_1||_1)=\{s\} \subseteq \{v, s\}=||\mathit{Ref}CV_1||_1$, we also have that $\M, s \models \mathit{Def}CV_1 \supp \mathit{Ref}CV_1$. Hence, $\M, v \models B(\mathit{Def}CV_1 \supp \mathit{Ref}CV_1)$ and $\M, v \models T_{\mathit{Def}CV_1} ( \mathit{Ref}CV_1)$.
Now, since $\mathit{most}(||\mathit{Def}CV_1 \land \mathit{Ref}CV_1||_1)=\{v, s\} \not\subseteq \{w, s\}=||\neg \mathit{Good}V_1||_1$, we have that $\M, s \not\models (\mathit{Def}CV_1\land\mathit{Ref}CV_1) \supp \neg\mathit{Good}V_1$. Hence, $\M, v \models \neg B((\mathit{Def}CV_1\land\mathit{Ref}CV_1) \supp \neg\mathit{Good}V_1)$. Finally, we have $\M, v \models \neg T_{(\mathit{Def}CV_1\land\mathit{Ref}CV_1)} (\neg\mathit{Good}V_1)$.

\end{example}

We now examine the relation between $\F$ and $\supp$. Theorem~\ref{teo:connectionF} has already highlighted a syntactic connection (from $\F$ to $\supp$ via the translation $^\ast$ in Def.~\ref{def:translation}). Here, by using their semantics, we uncover a stronger tie. Recall that $\F$ is sound and complete w.r.t.\ all preference models $\langle S, \succeq, V \rangle$ 
which fulfil the limitedness condition, see~\cite{Xav2015}. %This condition is part of our semantics as well. 
%As shown below, there is a strong connection between Trust models and preference models for $\F$. 
We denote by $\models^\F$ the semantical consequence relation in $\F$. 
%We have:
%Indeed, we can show that the semantics {for the flat fragment of} $\F$ align with our semantics for support. 
%Under the assumption of Soundness and Completeness for \logic, as shown in the next section, this implies that the axiomatization for $\supp$ coincides with a flat fragment of the axiomatization of $\F$.
%
%\dom{Maybe we should delete the next part since we already mentioned it now at the beginning of this section:}
%
%The main differences between Trust models and preference models for $\F$ are the following. 
%
%We only consider a flat fragment since nesting the support operator does not allow us to capture any meaningful trust scenario.\footnote{The nesting of support operators becomes more interesting when one also allows for the nesting of belief into support.} Furthermore, we allow for multiple preference relations to occur; this addition is crucial since it avoids $\supp$ being an absolute operator, allowing our agent to consider different support systems in the same model. Finally, we introduce the additional relation $R$, representing the belief relation among possible worlds.
%\end{remark}

% As shown in Theorem~\ref{teo:connectionF}, all axioms of $\F$ without nested modal operators are derivable in our axiomatization of $\supp$.  The theorem establishes that the axioms and rules of our support operator axiomatize the flat fragment of $\F$. More precisely 

\begin{theorem}\label{teo:flat-fragmentF}
    For any set of formulas $\Gamma$ and formula $\alpha$ in the language of $\F$ that do not contain nested modal operators, we have:
    $\Gamma \models^\F \alpha  \Leftrightarrow   \Gamma^* \models \alpha^*$.

%\begin{equation*}
%    \Gamma \models^\textbf{F} \alpha  \Leftrightarrow   \Gamma^* \models \alpha^*
%\end{equation*}
\end{theorem}

% To prove this theorem we are going to prove the following lemma.

% \begin{lemma}
%     Given a preference model $\mathcal{M}=\langle W, \succeq, V \rangle$, fulfilling the limitedness condition, a state $w \in W$ and a formula $\varphi$ in the language of \textbf{F} which does not contain nested modal operators such that $\M, w \models \varphi$ then given the Trust model  $\mathcal{M'}:=\langle S', S_1, \succeq_1, V', R \rangle$, with $S':=S_1:=W$, $\succeq_1:=\succeq$, $V':=V$ and $R:= S\times S$ it holds that $ \M', w \models \varphi^*$. 

% Furthermore given a Trust model  $\mathcal{N}:=\langle S, (S_i)_{i \in I}, (\succeq_i)_{i \in I}, R, V \rangle$ and a state $s \in S_i$  such that $\M, s \models \varphi^*$ then given the preference model  $\mathcal{N'}:=\langle S_i, \succeq_i, V \rangle$, it holds that $ \mathcal{N'}, w \models \varphi^*$. 
% \end{lemma}

\begin{proof}

Both directions proceed by contraposition. 

    $(\Rightarrow)$ We show that given a Trust model invalidating the semantical consequence for support we can find a preference model invalidating the semantical consequence for $\F$.
    Assume that $\Gamma^* \not\models \alpha^*$. Hence, there exists a Trust model $\mathcal{M}=\langle S, (S_i)_{i \in I}, (\succeq_i)_{i \in I}, R, V \rangle$ and a state $s \in S_i$ such that $\forall \beta^* \in \Gamma: \M, s \models \beta^*$ and $ \M, s \not\models \alpha^*$. We cut down the Trust model into a preference model as follows $\mathcal{M'}:=\langle S_i, \succeq_i, V \rangle$. By definition, $\M'$ is a preference model fulfilling the limitedness condition.\footnote{In \cite{Xav2015}, the limitedness condition is stated for every formula of $\F$. Since the truth sets of obligations and modalities are those of $\top$ or $\bot$, our limitedness condition is equivalent to the one given in $\F$.} Observe that no formula in $\Gamma^* \cup \{\alpha^* \}$ contains the operator $B$. Hence, the evaluation of the formulas in $\Gamma^* \cup \{\alpha^* \}$ at the state $s \in S_i$ coincides with the evaluation of the formulas in $\Gamma \cup \{\alpha\}$ in $\M$. This lets us conclude $\forall \beta \in \Gamma: \M', s \models \beta$ and $\M', s \not \models \alpha$. 

    $(\Leftarrow)$  Given a preference model 
    invalidating $\Gamma \models^{\F} \alpha$, we provide a Trust model invalidating $\Gamma^* \models \alpha^*$. Assume to have
    %From $\Gamma \not\models^{\F} \alpha$ \agata{follows the existence} of 
    the preference model $\mathcal{M}=\langle S, \succeq, V \rangle$, fulfilling the limitedness condition and a state $s \in S$ such that $\forall \gamma \in \Gamma: \M, s \models \gamma$ and $ \M, s \not\models \alpha$. We extend it into a Trust model with only one element in the partition, as follows $\mathcal{M'}:=\langle S, (S_i)_{i\in I}, (\succeq_i)_{i\in I}, V, R \rangle$, with $I:=\{1\}$, $S_1:=S$, $\succeq_1:=\succeq$ and $R:= S\times S$. By definition, $\M'$ is a Trust model. Again no formula in $\Gamma^* \cup \{\alpha^* \}$ contains the operator $B$. Hence the evaluation of the formulas in $\Gamma^* \cup \{\alpha^* \}$ at the state $s \in S$ coincides with the evaluation of the formulas in $\Gamma \cup \{\alpha\}$ in $\M$. This lets us conclude $\forall \beta^* \in \Gamma^*: \M', s \models \beta^*$ and $ \M', s \not \models \alpha^*$. 
\end{proof}

By the Soundness and Completeness of \logic\ w.r.t. Trust models, proved in the next section, it follows that the axioms and rules of $\supp$ axiomatize the flat fragment of $\F$.
% As shown in Theorem~\ref{teo:connectionF}, all axioms of $\F$ without nested modal operators are derivable in our axiomatization of $\supp$.  The theorem establishes that the axioms and rules of our support operator axiomatize the flat fragment of $\F$. More precisely 

%\dom{Add comment of ABS B}

% \begin{proof}
%     $"\Rightarrow"$  We prove this direction via contraposition. Therefore we assume for a given $\Gamma$ and $\varphi$ that $\Gamma^* \not\models \varphi^*$. Hence there exists a Trust model $\mathcal{M}=\langle S, (S_i)_{i \in I}, (\succeq_i)_{i \in I}, R, V \rangle$, an $i \in I$ and a state $s \in S_i$ such that $\forall \gamma \in \Gamma: \M, s \models \gamma^*$ and $ \M, s \not \models \varphi^*$. Observe that no formula in $\Gamma* \cup \{\varphi^* \}$ contains the operator $B$. This means that the evaluation of the formulas in $\Gamma* \cup \{\varphi^* \}$ at the state $s$ only depends on the preference model $\mathcal{M'}:=\langle S_i, \succeq_i, V \rangle$. Hence, $\forall \gamma \in \Gamma: \M', s \models \gamma^*$ and $ \M', s \not \models \varphi^*$. Because the semantics of \textbf{F} coincide with the semantics of our Trust models given the translation $*$ we can conclude that $\Gamma^* \not\models^\textbf{F} \varphi^*$.

%     $"\Leftarrow"$ This follows from the fact that all axioms and rules of \textbf{F} are sound in our semantics.
% \end{proof}

% Before moving on to the soundness and completeness proof of the system we should mention that we just defined a flat fragment of \AA quvist's system $\F$. 

\section{Soundness, completeness and complexity of \logic}
\label{sect:res}

%\subsection{\dom{Soundness and completeness}}\label{sect:compl}

We start with the soundness of \logic\ w.r.t. Trust models.

\begin{theorem}[Soundness]\label{teo:soundness}
  Given $\Phi \subseteq \ForT$ and $\alpha \in \ForT$ it holds that:  $\Phi \vdash \alpha \Rightarrow  \Phi \models \alpha$.
\end{theorem}

\begin{proof}
    Proceed, as usual, by induction of the length of the derivation. We distinguish cases according to the last rule applied, showing the details of the cases for
    %showing the soundness of each axiom and rule. For the sake of brevity, we will only consider 
    the axiom \textbf{LL+} and the rule $\mathbf{S5_F}$.

    \textbf{LL+}: Given a Trust model $\mathcal{M}=\langle S, (S_i)_{i \in I}, (\succeq_i)_{i \in I}, R, V \rangle$ and a state $s \in S_i$ such that $\M,s \models \neg (\varphi \IFF \psi) \supp \bot$, then we get $\mathit{most}(||\neg (\varphi \IFF \psi)||_i) \subseteq \emptyset$. Given the limitedness assumption, this is equivalent to $||\neg (\varphi \IFF \psi)||_i = \emptyset$ and furthermore to $||\varphi \IFF \psi||_i = S_i$. Hence, $\varphi$ and $\psi$ are equivalent in every state of $S_i$. Therefore the sets $\mathit{most}(||\varphi||_i)$ and $\mathit{most}(||\psi||_i)$ coincide, i.e., $\M,s \models (\varphi \supp \chi) \IFF (\psi \supp \chi)$.

    $\mathbf{S5_F}$: Given a Trust model $\mathcal{M}=\langle S, (S_i)_{i \in I}, (\succeq_i)_{i \in I}, R, V \rangle$, we assume $((\neg)(\varphi_1 \supp \psi_1) \land \dots \land (\neg)(\varphi_n \supp \psi_n))\rightarrow \chi$ to be true in every state of $\mathcal{M}$. Given a state $s \in S_i$ such that $\M, s \models ((\neg)(\varphi_1 \supp \psi_1) \land \dots \land (\neg)(\varphi_n \supp \psi_n))$ holds, it follows that $\forall w \in S_i$ $\M, w \models ((\neg)(\varphi_1 \supp \psi_1) \land \dots \land (\neg)(\varphi_n \supp \psi_n))$ because by virtue of the semantics of $\supp$, all the states of a given partition class satisfy the same (negated) support formulas.    %because neither a support formula nor its negation depend on the structure of a particular state and only on the relation $\succeq_i$. 
    Therefore, we have that $\forall w \in S_i$ $\M, w \models \chi$, which means $||\neg \chi||_i = \emptyset$ and finally $\M, s \models \neg \chi \supp \bot$.
\end{proof}

Completeness is shown via the canonical model construction, adapted to our framework from the technique outlined in \cite{grossi2022reasoning}. The needed modifications are the following. First, we have to ensure that \logic\ allows us to derive all the axioms and rules required for the construction to proceed. Furthermore, unlike the models considered in~\cite{grossi2022reasoning}, Trust models incorporate multiple preference frames, the belief operator $B$, and include the limitedness condition.
%, which posed a significant challenge to address.

The modifications are implemented as follows. The required axioms and rules for their proof to go through are 
those of $\F$ without $\mathbf{D^*}$ (which corresponds to limitedness). In Sect.~\ref{sec:axiom} we have shown that with the exception of \textbf{5} and \textbf{Abs} the axioms of $\F$ are derivable in \logic. 
We prove that we can derive all the necessary properties of the canonical model even in the absence of axioms \textbf{5} and \textbf{Abs}, by relying on other rules of \logic, primarily on $\bf S5_F$. The multiple preference frames are handled by partitioning the maximal consistent sets used in the canonical model construction into equivalence classes containing the same support formulas. The addition of belief is easy: We equip the canonical model with the accessibility relation in the usual Kripke fashion. 
Incorporating the limitedness condition poses the challenge of guaranteeing that, in every preference model of our canonical model, each non-empty set $||\varphi||_i$ contains a maximal element according to the preference relation. We address this by using the axiom \textbf{ST}.
 %
 %The addition of the limitedness condition poses the challenge of ensuring that each non-empty set of maximal consistent sets in each preference frame contains a maximal element according to the preference relation. This can be achieved by utilizing the axiom \textbf{ST}.

% Completeness is shown via the canonical model construction, by adapting the technique in \cite{grossi2022reasoning} to our framework. The adjustments are as follows.  First, we make sure that our fragmented language still allows us to derive all the needed inferences. Second, account for the partitioning of the frames by partitioning the maximal consistent sets of formulas accordingly. Last, we incorporate the limitedness condition and the belief operator $B$, neither of which are featured in the aforementioned paper; 
% \agata{the use of limitednedd... bla}

%\footnote{no shallow, no limitendess, no partitions and no beliefs}. 

\begin{definition}\label{def:MCS}
    A set $\Gamma \subseteq \ForT$ is called a maximal consistent set (MCS for short) if (a) $\Gamma \not\vdash \bot$, and (b) For every $\alpha \in \ForT$ either $\alpha \in \Gamma$ or $\neg \alpha \in \Gamma$.
\end{definition}

Although not all the states in a model validate the same support formulas, we still need to make sure that all the states inside the same preference frame $S_i$ do. %The heavy lifting in this regard will mostly be achieved by the rule 2 and the axiom of factivity.
For that reason, we partition the maximal consistent sets into equivalence classes, containing the same support formulas. We also define a set $\supp_\varphi(\Gamma)$ containing all formulas which are supported in a MCS $\Gamma$ by a formula $\varphi$.

\begin{definition}\label{def:supp-form}
    Given $\Gamma, \Delta \subseteq \ForT$ and $\varphi \in \ForCL$ we define:
\begin{itemize}
    \item $\Gamma^\supp := \{\chi \supp \psi \in \Gamma\} \qquad \qquad \hspace{11pt}$
%\end{itemize}
(We write  $\Gamma \leftrightsquigarrow \Delta$ if $\Gamma^\supp = \Delta^\supp$). 
%\begin{itemize}
    \item $\supp_\varphi(\Gamma) := \{\psi: \varphi \supp \psi \in \Gamma\} \qquad$
(We call $\Delta$  $\varphi$-likely for $\Gamma$ if $\supp_\varphi(\Gamma)  \subseteq \Delta$). 
\end{itemize}
\end{definition}

\begin{fact}
        $\leftrightsquigarrow$ is an equivalence relation on the set of all MCSs. We write $[\Gamma]_\leftrightsquigarrow$ for the equivalence class containing $\Gamma$.
\end{fact}

Each equivalence class serves as a basis for a preference frame in our canonical model. The maximal consistent sets, which are $\varphi$-likely for $\Gamma$ are our candidates for the most likely $\varphi$ states in the preference frame based on the equivalence class $[\Gamma]_\leftrightsquigarrow$, as they contain all formulas supported by $\varphi$.

Before
moving forward, we need a result that will be used repeatedly.
It asserts that if $\psi$ is not supported by $\varphi$ in $\Gamma,$ then we can construct a MCS $\Delta$ with the same support formulas as $\Gamma$, and including the negation of $\psi$ as well as all the propositions supported by $\varphi$.

%that if a formula $\varphi$ does not appear in the head of a support formula in a MCS $\Gamma$, then we can construct a MCS with the same support formulas that contains all the heads but invalidates $\varphi$.

\begin{lemma}\label{lem:extension}
   Given a MCS $\Gamma$ and a propositional formula $\psi$ with $\psi \not \in \supp_\varphi(\Gamma)$, then there exists $\Delta \in [\Gamma]_\leftrightsquigarrow$ such that $\{\neg \psi\} \cup \supp_\varphi(\Gamma) \subseteq \Delta$.
\end{lemma}

\begin{proof}
We show the consistency of the set $A:=\{\neg \psi\} \cup \supp_\varphi(\Gamma) \cup \Gamma^\supp \cup \{\neg(\chi \supp \gamma): \chi \supp \gamma \not\in \Gamma\}.$  If this holds, we can extend the set to an MCS $\Delta$ which, by construction, is contained in $[\Gamma]_\leftrightsquigarrow$. We prove the consistency of $A$ by contradiction. Assume $A \vdash \bot$. Hence, we can find $$\varphi_1, ..., \varphi_n \in \supp_\varphi(\Gamma), \pi_1 \supp \psi_1, ..., \pi_m \supp \psi_m \in \Gamma^\supp$$ and $\neg(\chi_1 \supp \gamma_1), ..., \neg(\chi_k \supp \gamma_k) \in \Gamma$ such that $\alpha \land \varphi_1 \land ... \land \varphi_n \land \neg \psi \vdash \bot$ with $$\alpha:= \pi_1 \supp \psi_1 \land ... \land \pi_m \supp \psi_m \land \neg(\chi_1 \supp \gamma_1) \land ... \land \neg(\chi_k \supp \gamma_k).$$

The $\mathbf{CL}$ axioms and the deduction theorem yield $\vdash \alpha \rightarrow ((\varphi_1 \land \dots \land \varphi_n) \rightarrow \psi).$
From $\mathbf{S5_F}$ we get $$ \vdash \alpha \rightarrow (\neg((\varphi_1 \land ... \land \varphi_n) \rightarrow \psi) \supp \bot)$$ and then $\vdash \alpha \rightarrow (\varphi \supp ((\varphi_1 \land ... \land \varphi_n) \rightarrow \psi))$ with the help of \textbf{Nec}. This leads to  $\psi \in \supp_\varphi(\Gamma)$ because of \textbf{RW}, a contradiction to our assumption.
\end{proof}

We define the states and preference relations $\succeq_\Gamma$ for our canonical model (the index $
\Gamma$ is a representative of an equivalence class of the equivalence relation $\leftrightsquigarrow$). The states are $(\Delta, \varphi, i)$ where $\Delta$ is a MCS in the equivalence class $[\Gamma]_\leftrightsquigarrow$,  $i \in \{0, 1, 2\}$, and $\varphi$ is a propositional formula. 
%The formula $\varphi$ and the integer $i$ 
 $\varphi$ and $i$
are used to pinpoint the maximal $\varphi$-states according to the relation $\succeq_\Gamma$ and to ensure that the maximal elements of $\succeq_\Gamma$ coincide with the states satisfying the supported formulas within $\Gamma$, see 
Cor.~\ref{cor:iff}.

For a MCS $\Gamma$ we use the following notation:
$$S_\Gamma:= [\Gamma]_\leftrightsquigarrow \times \ForCL \times \{0,1,2\}
\quad \mbox{ and } \quad
[\delta]_\Gamma:= \{(\Delta, \varphi, i) \in S_\Gamma: \delta \in \Delta\}.$$

\begin{definition}\label{PrefRel}
  %  For a MCS $\Gamma$ we define $S_\Gamma:= [\Gamma]_\leftrightsquigarrow \times \ForCL \times \{0,1,2\}$ and $[\delta]_\Gamma:= \{(\Delta, \varphi, i) \in S_\Gamma: \delta \in \Delta\}$.
    The preference relation $\succeq_\Gamma \subseteq S_\Gamma \times S_\Gamma$ is defined as follows:
    
    $(\Delta, \varphi, i) \succeq_\Gamma (\Omega, \psi, j)$ holds if and only if at least one of the following conditions holds:
    \begin{itemize}
        %\item $\Delta$ is $\varphi$-acceptable for $\Gamma$ and $\psi \not\in \Delta$
        \item $\Delta$ is $\varphi$-likely for $\Gamma$ and $\varphi \in \Omega$
        %\item $\Delta=\Omega$ and $\varphi=\psi$ 
        \item ($i=1$ and $j=0$) or ($i=2$ and $j=1$) or ($i=0$ and $j=2$)
        %\item $i=2$ and $j=1$
        %\item $i=0$ and $j=2$
    \end{itemize}
\end{definition}

After having defined our preference relations we show that the maximal elements in $[\delta]_\Gamma$ are $\delta \text{-likely for } \Gamma$. This is what we are aiming for as we want all the elements in $\mathit{max}([\delta]_\Gamma)$ to fulfil every formula supported by $\delta$ in $\Gamma$. Furthermore, if an MCS $\Delta \in S_\Gamma$ is $\delta \text{-likely for } \Gamma$ then $(\Delta, \delta, i)$ is a maximal element in $[\delta]_\Gamma$. These results will be the core of our completeness proof since they can be used to show that a support formula is true in a state of $S_\Gamma$ if and only if the formula appears in $\Gamma$. To prove this, we start proving the following technical lemma that establishes a connection between an $(\Delta, \varphi, i)$ appearing in $\mathit{max}([\delta]_\Gamma)$ and the MCS $\Gamma$.

\begin{lemma}
\label{lemma:connection}
   Given $(\Delta, \varphi, i) \in \mathit{max}([\delta]_\Gamma)$ then:
\[\begin{array}{ll}
 (a) & \Delta \textit{\ is\ } \varphi\textit{-likely\ for\ } \Gamma; \\ 
 (b) & \neg (\delta \rightarrow \varphi) \supp \bot \in \Gamma.
 % \item $ (\Delta, \varphi, i) \in max([\delta]_\Gamma) \text{ then } \Delta \text{ is } \delta \text{-common for } \Gamma$
  %\item $\text{ If } \Delta \text{ is } \delta \text{-common for } \Gamma \text{ then } (\Delta, \delta, i) \in max([\delta]_\Gamma)$
\end{array}\]
\end{lemma}

\begin{proof}
We assume $i=0$. The other cases being similar.

(a) Given $(\Delta, \varphi, 0) \in \mathit{max}([\delta]_\Gamma)$, we have $(\Delta,  \varphi, 1) \succeq_\Gamma (\Delta, \varphi, 0)$ by construction of $\succeq_\Gamma$. This implies $(\Delta, \varphi, 0) \succeq_\Gamma (\Delta, \varphi, 1)$ by maximality. The latter only holds if $\Delta$ is $\varphi$-likely for $\Gamma$, and $\varphi \in \Delta$ since no other condition applies.

%\text{ and } \Delta \text{ is NOT } \varphi \text{-common for } \Gamma$. $(\Delta, \varphi, 0) \not\in max([\delta]_\Gamma)$. Therefore  $(\Delta, \varphi, i) \in max([\delta]_\Gamma)$ implies $\Delta \text{ is } \varphi \text{-common for } \Gamma$. \\

(b) Assume that there exists a MCS $\Omega \in S_\Gamma$ such that $\delta \in \Omega$ but $\varphi \not\in \Omega$. For $\Omega$ we then have $(\Delta, \varphi, 0) \not\succeq_\Gamma (\Omega, \delta, 1)$, but $(\Omega, \delta, 1) \succeq_\Gamma (\Delta, \varphi, 0)$ which is a contradiction to $(\Delta, \varphi, i) \in \mathit{max}([\delta]_\Gamma)$. Hence, we can infer that such an MCS $\Omega \in S_\Gamma$ does not exist, which means $[\delta]_\Gamma \subseteq [\varphi]_\Gamma$. In other words, every MCS in $S_\Gamma$ contains the formula $\delta \rightarrow \varphi$. More specific each MCS $\Pi$ with $\Pi \in [\Gamma]_\leftrightsquigarrow$ contains the formula $\delta \rightarrow \varphi$, which means $[\Gamma]^\supp \cup \{\neg (\delta \rightarrow \varphi)\}$ is inconsistent. This lets us infer $[\Gamma]^\supp \vdash \delta \rightarrow \varphi$, hence we can find finitely many support formulas in $\Gamma^\supp$ such that $(\varphi_1 \supp \psi_1 \land \dots \land \varphi_n \supp \psi_n)  \vdash \delta \rightarrow \varphi$. %\dom{NOTE TO SELF: refine argument} 
By applying the deduction theorem and $\mathbf{S5_F}$ we get $\vdash (\varphi_1 \supp \psi_1 \land \dots \land \varphi_n \supp \psi_n)  \rightarrow (\neg (\delta \rightarrow \varphi) \supp \bot)$ and finally $\neg (\delta \rightarrow \varphi) \supp \bot \in \Gamma$.
\end{proof}

\begin{corollary}\label{cor:mac-com}
  Given $(\Delta, \varphi, i) \in S_\Gamma$ then
\[\begin{array}{ll}
  (a) & (\Delta, \varphi, i) \in \mathit{max}([\delta]_\Gamma) \text{\ implies\ } \Delta \text{\ is\ } \delta \text{-likely\ for\ } \Gamma; \\
  (b) & \Delta \text{\ being\ } \delta \text{-likely\ for\ } \Gamma \text{\ implies\ } (\Delta, \delta, i) \in \mathit{max}([\delta]_\Gamma).
\end{array}\]
\end{corollary}

\begin{proof}

(a) By $ (\Delta, \varphi, i) \in \mathit{max}([\delta]_\Gamma)$ and Lemma~\ref{lemma:connection} we derive $\neg (\delta \rightarrow \varphi) \supp \bot \in \Gamma$. Using the  derivable axiom $\mathbf{K_\Box}$ and Necessitation for $\Box$ (see Theorem~\ref{teo:connectionF}) we get $\neg (\delta \leftrightarrow (\delta \land \varphi)) \supp \bot \in \Gamma$ (using the classical tautology  $(\delta \rightarrow \varphi)  \rightarrow (\delta \leftrightarrow (\delta \land \varphi))$). Let us take an arbitrary $\gamma \in \supp_\delta(\Gamma)$, this means $\delta \supp \gamma \in \Gamma$. Since $\neg (\delta \leftrightarrow (\delta \land \varphi)) \supp \bot \in \Delta$ we can apply  \textbf{LL+} to derive $(\delta \land \varphi) \supp \gamma \in \Delta$. Furthermore, by applying \textbf{SH}, we get $\varphi \supp \delta \rightarrow \gamma \in \Delta$. By 
Lemma~\ref{lemma:connection} we know that  $\Delta \text{ is } \varphi \text{-likely for } \Gamma$, which means $\supp_\varphi(\Gamma) \subseteq \Delta$ and therefore $ \delta \rightarrow \gamma \in \Delta$. By assumption, we have  $\delta \in \Delta$, which lets us conclude  $\gamma \in \Delta$. Since $\gamma$ was arbitrary we get $ \supp_\delta(\Gamma) \subseteq \Delta$. 

(b) Since $\Delta$ is $\delta$-likely for $\Gamma$ by axiom {\bf ID} we have $\delta \in \supp_\delta(\Gamma) \subseteq \Delta$ and therefore $\Delta \in [\delta]_\Gamma$.
%we start out with the fact that $\Delta \in [\delta]_\Gamma$ because of the identity axiom. 
If we take an arbitrary MCS $\Omega \in S_\Gamma$ with $\Omega \in [\delta]_\Gamma$ and an arbitrary propositional formula $\pi$ we end up with $(\Delta, \delta, i) \succeq_\Gamma (\Omega, \pi, j)$ since the first point in the definition of $\succeq_\Gamma$ is fulfilled. 
\end{proof}

% \begin{corollary}
%        Given a MCS $\Gamma$ and a propositional formula $\varphi$ with $\varphi \not \in \supp_\psi^{-1} \Gamma$, then there exists $(\Delta, \psi, i) \in S_\Gamma$ with $(\Delta, \psi, i) \in max([\psi]_\Gamma)$ and $\neg \varphi \in \Delta$.
% \end{corollary}

We can finally show that our construction works as intended, namely that every formula $\psi$ supported by a formula $\delta$ according to a MCS $\Gamma$ is contained in all the maximal $\delta$ states in the equivalence set of $\Gamma$.

\begin{corollary}\label{cor:iff}
       Given a MCS $\Gamma$ and two propositional formulas $\delta$ and $\psi$, it holds that
       $$\delta \supp \psi \in \Gamma  \text{ if and only if } \; \forall \, (\Delta, \varphi, i) \in \mathit{max}([\delta]_\Gamma): \psi \in \Delta $$
\end{corollary}
\begin{proof}
$(\Rightarrow)$ Given $(\Delta, \varphi, i) \in \mathit{max}([\delta]_\Gamma)$ we can derive that $\Delta \text{ is } \delta \text{-likely for } \Gamma$ via Corollary~\ref{cor:mac-com}. By assumption we get $\psi \in \supp_\delta(\Gamma) \subseteq \Delta$. 

$(\Leftarrow)$ By contraposition. Assume $\psi \not\in \supp_\delta(\Gamma)$. By Lemma~\ref{lem:extension} we infer that there exists a MCS $\Delta \in [\Gamma]_\leftrightsquigarrow$ such that $\{\neg \psi\} \cup \supp_\delta(\Gamma) \subseteq \Delta$. By construction $\Delta$ is $\delta \text{-likely for } \Gamma$. Corollary~\ref{cor:mac-com} gives us $(\Delta, \varphi, i) \in \mathit{max}([\delta]_\Gamma)$. Since $\neg\psi \in \Delta$ we get $\psi \not\in \Delta$ by consistency.
\end{proof}

% \begin{corollary}
%        Given a MCS $\Gamma$ and two propositional formulas $\varphi$ and $\psi$ then the following equivalence holds:
       
%        $$\psi \in \supp_\varphi^{-1} \Gamma  \text{ if and only if } \forall (\Delta, \delta, i) \in max([\varphi]_\Gamma): \psi \in \Delta $$
       
% \end{corollary}
% \begin{proof}
% "$\Rightarrow:$" Given $(\Delta, \delta, i) \in max([\varphi]_\Gamma)$ we can derive that $\Delta \text{ is } \varphi \text{-common for } \Gamma$ via Corollary 1. By assumption we get $\psi \in \supp_\varphi^{-1} \Gamma \subseteq \Delta$. \\

% "$\Leftarrow:$" We are going to prove this direction by contraposition. Hence, let us assume $\psi \not\in \supp_\varphi^{-1} \Gamma$. Using lemma 2 we derive that there exists a MCS $\Delta \in [\Gamma]_\leftrightsquigarrow$ such that $\{\neg \psi\} \cup \supp_\varphi^{-1} \Gamma \subseteq \Delta$. By construction $\Delta$ is $\varphi \text{-common for } \Gamma$. We can use Corollary 1 to obtain  $(\Delta, \varphi, i) \in max([\varphi]_\Gamma)$. Since $\neg\psi \in \Delta$ we get $\psi \not\in \Delta$ by consistency.
% \end{proof}

Now, we proceed to define the canonical model. To begin, let us fix $I$ as a set consisting of one representative of each equivalent class of $\leftrightsquigarrow$. 

\begin{definition}[Canonical model]\label{CanMod}
Let $\mathcal{M}^{Can}:=\langle S, (S_\Gamma)_{\Gamma \in I}, (\succeq_\Gamma)_{\Gamma \in I}, R, V \rangle$ where:
\begin{itemize}
    \item $S:=\bigcup_{\Gamma \in I} S_\Gamma$;
    \item $V(p):= \{(\Delta, \varphi, i) \in S: p \in \Delta \}$;
    \item $\succeq_\Gamma \subseteq S_\Gamma \times S_\Gamma$ is defined as in Definition~\ref{PrefRel};
    \item $R \subseteq S \times S$ is defined as $(\Delta, \varphi, i) R (\Omega, \psi, j)$ if for all $\alpha \in \ForT: ~ (B(\alpha) \in \Delta \Rightarrow \alpha \in \Omega)$.
\end{itemize}

\end{definition}

We now begin the final steps of our completeness proof. First, we present the truth lemma, which states that a formula is true at a state in the canonical model if and only if the formula is an element of the maximal consistent set of this state. Following this, we demonstrate that $\mathcal{M}^{Can}$ is a Trust model.

\begin{lemma}[Truth lemma]\label{lem:truth-lemma}
    $\mathcal{M}^{Can}, (\Delta, \pi, i) \models \alpha$ iff $\alpha \in \Delta$.
\end{lemma}

\begin{proof}
 Proceeds, as usual,  by structural induction on the formula $\alpha$. See appendix.
\end{proof}

\begin{lemma}\label{lem:can-trust-model}
    $\mathcal{M}^{Can}$ in Definition~\ref{CanMod} is a Trust model.
\end{lemma}

\begin{proof}
First, we prove that for each $\Gamma \in I$ $\langle S_\Gamma , (\succeq_\Gamma), V \rangle$ fulfils the limitedness condition.
Let $\varphi \in \ForCL$ and $\Gamma \in I$ with $[\varphi]_\Gamma\not=\emptyset$. Hence there exists a MCS $\Delta \in [\Gamma]_\leftrightsquigarrow$ with $\varphi \in \Delta$. \textbf{ST} tells us that $\neg(\varphi \supp \bot) \in \Delta$. This implies $\varphi \supp \bot \not\in \Delta$ and finally $\bot \not\in \supp_\varphi(\Gamma)$. Given that the set $\supp_\varphi(\Gamma) $ is closed under consequences because of \textbf{RW}, we can conclude that it is consistent. By 
Lemma~\ref{lem:extension}, we can now extend $\supp_\varphi(\Gamma)$  to a MCS $\Pi \in [\Gamma]_\leftrightsquigarrow$. By construction $\Pi$ is $\varphi \text{-likely for } \Gamma$. Using part (b) of Corollary~\ref{cor:mac-com} we obtain $(\Pi, \varphi, i) \in \mathit{max}([\varphi]_\Gamma)$, which makes $\mathit{max}([\varphi]_\Gamma)$ non-empty. We use Lemma~\ref{lem:truth-lemma} to conclude $||\varphi||_\Gamma \not = \emptyset \Rightarrow \mathit{most}(||\varphi||_\Gamma) \not = \emptyset$. Since $\varphi$ and $\Gamma$ were arbitrary we are done.

Furthermore, we need to show that the relation $R$ is transitive and serial. We will, from now on, write $\Delta R \Omega$ for $(\Delta, \varphi, i) R (\Omega, \psi, j)$ since $R$ only depends on the sets.
{\em Transitivity}: Assume $\Delta R \Omega$ and $\Omega R \Pi$. %If for any $\alpha$ it holds that 
From $B(\alpha) \in \Delta$ by \textbf{4B} we derive $B(B(\alpha)) \in \Delta$. By assumption $\Delta R \Omega$ and the construction in Definition~\ref{CanMod} $B(\alpha) \in \Omega$; the same argument applies for $\Omega R \Pi$ to derive $\alpha \in \Pi$. Since $\alpha$ was arbitrary, we can conclude $\Delta R \Pi$.
{\em Seriality}: Given a MCS $\Delta$ we get $B(\top) \in \Delta$ via the rule of necessitation for $B$. This implies $B(\bot) \not\in \Delta$ because of the axiom \textbf{DB} and the consistency of $\Delta$. Now we take a look at the set $A:=\{\alpha: B(\alpha) \in \Delta\} $. Because of the axiom $\mathbf{KB}$, we know that $A$ is closed under consequences, and since $\bot \not\in A$, we also know $A$ to be consistent. We can, therefore, extend it to a maximal consistent set $\Pi$. By definition $\Delta R \Pi$, which makes $R$ serial.
%\Rightarrow most(||\varphi||_i)\not=\emptyset$
\end{proof}

With these two lemmas, we are now ready to prove the completeness of \logic. This will be done in the usual manner by showing that for every non-derivable formula, there exists a state in our canonical model where it does not hold.

\begin{theorem}[Completeness]\label{teo:completness}
   Given $\Phi \subseteq \ForT$ and $\alpha \in \ForT$ it holds that:  $\Phi \models \alpha \Rightarrow  \Phi \vdash \alpha$.
\end{theorem}

\begin{proof}
   % "$\Rightarrow$" Follows from the soundness theorem.
%Proceed by contraposition. 
By contraposition. From $\Phi \not\vdash \alpha$ follows that $\Phi \cup \{\neg \alpha\}$ is consistent and can therefore be extended to a MCS $\Delta$.  By Lemma~\ref{lem:truth-lemma} every formula in $\Delta$ holds in the canonical model in a state of the form $(\Delta, \gamma, i)$. Hence $\forall \beta \in \Phi: (\Delta, \gamma, i)\models \beta$ and $(\Delta, \gamma, i) \not\models \alpha$, which gives us $\Phi \not\models \alpha$ by Lemma~\ref{lem:can-trust-model}.
\end{proof}

We now discuss the complexity results for \logic. We split the problem into two parts: $(i)$ we reduce \logic\ by ignoring its support part and focusing on the Boolean and belief parts, and then $(ii)$ we reinstate the support part, completing the proof.

\begin{definition}[\logic-reduction]
    The \logic-reduction is obtained by reducing $\ForT$ to $\ForT'$ and transforming a model $\mathcal{M} := \langle S, \_, \_, R, V \rangle$ into a model $\mathcal{M}' := \langle S, R, V'\rangle$ in the following way:
    \begin{itemize}
        \item $\varphi$ formulas of $\ForT$ remain unchanged in $\ForT'$;
        \item All $\alpha\in\ForT$ of the form $\varphi\supp\varphi$, are mapped to novel propositional variables taken from a set $\mathit{Prop}'$, where $\mathit{Prop}' \bigcap \mathit{Prop} = \emptyset$;
        \item All the other $\alpha$ formulas are adjusted accordingly;
        \item $V'$ extends $V$ to also include in its domain $\mathit{Prop}'$, according to the following rule: if $\M, s \models \varphi \supp \psi $ and $p'$ is the propositional variable corresponding to $\varphi \supp \psi$, then $s\in V'(p')$.
    \end{itemize}
\end{definition}
%\marginpar{is the remark needed?}
%\begin{remark}
    %$\ForT'$ is a set of $\KDf$ formulas and $\mathcal{M}'$ a serial and transitive relational model.
%\end{remark}

%\begin{lemma}\label{decisionbelief}
    %The decision problem for the reduction of \logic\ is PSPACE-complete.
%\end{lemma}

%\begin{proof}
    %Since $\ForT'$ is a set of $\KDf$ formulas and the $\mathcal{M}'$s are serial and transitive relational models, the results in~\cite{Nguyen2004OnTC,Ladner1977TheCC,HALPERN1992319} hold also for $\ForT'$ formulas and, in turn, for the reduction of \logic.
%\end{proof}

\begin{theorem}\label{dec}
    The decision problem for \logic\ 
    %language presented in section~\ref{axiom} 
    is PSPACE-complete.
\end{theorem}

\begin{proof}
    First note that the decision problem for a \logic-reduction is PSPACE-complete. This follows from the fact that $\ForT'$ is a set of $\KDf$ formulas and the $\mathcal{M}'$s are serial and transitive relational models. Therefore the results given in~\cite{HALPERN1992319,Ladner1977TheCC,Nguyen2004OnTC} hold also for $\ForT'$ formulas and, in turn, for the \logic-reduction. Now, take the conjunction of all support formulas corresponding to the propositional variables of $\mathit{Prop}'$ that appear in $\ForT'$ and are mapped to true. To prove the theorem, we have to show that the problem of deciding those support formulas' satisfiability is within PSPACE.
    Refining a result in~\cite{Friedman_Halpern_small_model_constructions}, \cite{DimAIML} shows that the satisfiability problem for the (full) logic $\F$ is NP-complete.
    The result follows by Theorem~\ref{teo:flat-fragmentF}.
    %\agata{The result} follows from Theorem~\ref{teo:flat-fragmentF}, and~\cite{DimAIML}, that, refining a result in~\cite{Friedman_Halpern_small_model_constructions}, shows that the satisfiability problem for the (full) logic $\F$ is NP-complete.
    %\agata{(improving the prior complexity upper bound in~\cite{CiabattoniOPRR23})}.  
    %and the fact that NP$\subseteq$PSPACE.
\end{proof}

Finally, we state the complexity of the model checking problem for \logic, i.e., the problem of deciding whether a formula $\alpha$ is satisfied by a state $s \in S $ of a model $\mathcal{M}$. 
%$\in \mathcal{M}$.

\begin{theorem}
    Given a model $\mathcal{M} := \langle S, \_, \_, R, V \rangle$ and a formula $\alpha \in \ForT$, let $n$ be the number of states $s \in S$ and $r$ the number of pairs $sRv$ determined by the accessibility relation $R$. Let $k$ be the number of support formulas, $k'$ the number of belief modalities, plus the number of atomic propositions in $\alpha$, plus the number of connectives in the formula. Then, the complexity of the model checking decision problem is $O(k\cdot n^2 + (k+k')\cdot (n+r))$.
\end{theorem}

\begin{proof}
    %We will proceed in two stages. The first step will consider only the modal part of the formula $\alpha$ under evaluation, while the second will take into consideration also the support part.
    We use the splitting methodology: first considering only the modal part of the formula $\alpha$, and afterwards the support part.
    For the first stage, apply a \logic-reduction to $\mathcal{M}$ and $\alpha$. This takes at most $k+k'$-steps (the propositional formulas are left unchanged and each support formula is substituted with a novel propositional formula). After the reduction, what is left is a model checking problem for a modal formula within a pointed Kripke relational model. As is well-known, the complexity of this problem is $O((n+r)\cdot (k+k'))$, see, e.g., \cite{Halp03}.
    For the second stage, translate back the novel propositional formulas to their respective $k$ support formulas. Take the partition $S_i$ which contains the state $s$ in which we must evaluate the formula. To evaluate a support formula $\varphi \supp \psi$, we need to compute the two sets $\mathit{most}(||\varphi||_i)$ and $||\psi||_i$. The latter is straightforward (since each state was already labeled during the first stage). The former requires at most $n^2$-steps (we are assuming the worst case in which $||\varphi||_i = S_i = S$): compare each state in $S_i$ with all other states in $S_i$, keeping track of the states that are preferred to other states. This must be done for all $k$ support formulas, thus, the complexity of the whole procedure is $O(k \cdot n^2)$. This gives us the whole complexity of the model checking problem, which is $O(k\cdot n^2 + (k+k')\cdot (n+r))$.
\end{proof}

\section{Conclusions and Future Works}
\label{sect:con}
We have introduced \logic, a logical framework for reasoning about decision trust based on two pillar concepts, namely belief and support. For the latter, we defined a novel non-monotonic conditional operator, which axiomatizes the flat fragment of the logic $\F$ and is based on preference-semantics.

Because of the generality of the concepts above and the way in which they are combined to formalize trust, \logic\ can integrate elements from the different approaches mentioned in Sect.~\ref{sect:art} within a unified framework. More precisely, we do not need to indicate \textit{specific} necessary cognitive conditions for the emergence of trust. Instead, we provide a way to get trust out of the support that exists between different formulas, which can capture the influence of different factors on trust. We also maintain a reference to cognitive features by integrating the belief modality. 

%The example below shows how \logic can be used to face the limitations discussed in Example~\ref{ex:amazon}:

Following up the discussion initiated in  Example~\ref{ex:amazon}, we now illustrate how \logic\ can be used to combine different elements that contribute to establishing trust. 

%The example that follows displays how you might use \logic to formalize a complex scenario and how you can handle the limitations discussed in Example~\ref{ex:amazon}:

\begin{example}
%Continuing Example~\ref{ex:amazon}, 
Assume that to trust $\mathit{Good}V_i$, customer $C$ seeks to fulfil three conditions: i) a cognitive-based one; ii) a reputation-based one; iii) a policy-based one.
The cognitive-based condition could be captured by the notion of occurrence Trust (denoted by formula $\mathit{OccT}V_i$) given in~\cite{10.1093/jigpal/jzp077}, which depends on multiple cognitive features of the agents involved such as the goals of $C$, the ability and intentions of $V_i$, and the effects of the actions of $V_i$ on the goals of $C$. The reputation-based condition could be represented, e.g., by the proposition $\mathit{TopRating}V_{i,j}$, while the policy-based condition by proposition $\mathit{Auth}V_i$.
Then, the formula $T_{\Gamma}(\mathit{GoodV_i})$ will indicate that customer $C$ trusts $V_i$ as a good vendor for reason $\Gamma$, where $\Gamma$ stands for $(\mathit{OccT}V_i \land \mathit{TopRating}V_{i,j} \land \mathit{Auth}V_i)$.
\end{example}

%Adopting a model from a specific paradigm may only allow us to express one of the three conditions, 
As emphasized by this example, in \logic, we have the flexibility to express several different conditions that specific models cannot capture alone. This flexibility comes, however, at the expense of reduced deductive power.

While our framework is primarily designed to capture decision trust, we claim that it could be versatile enough to encompass other notions of trust. 
In particular, as discussed in~\cite{Gambetta1988-GAMTMA}, many alternative definitions of trust in heterogeneous domains, such as reliability trust~\cite{Gambetta1988-GAMCWT} in the setting of economy, are based on the use of explicit supportive information. In addition, it is interesting to note that our approach has strong similarities with the approach followed in argumentation-based formalizations of trust~\cite{10.1080/19462166.2014.881417,VILLATA2013541}.
%In the future, we will formally explore how \logic\ can be used to formalize other notions of trust and how it compares with other formalization approaches.
%, such as, e.g., argumentation-based trust~\cite{10.1080/19462166.2014.881417,VILLATA2013541}, or reliance trust~\cite{Gambetta1988-GAMCWT}. In fact, we are able to formalize every notion of trust whose computation is based on the use of explicit supportive information (see~\cite{Gambetta1988-GAMTMA} for a general discussion).} %In our framework, evidential sources would be formalized as formulas in the language and the support those sources provide for specific information would be captured by $\supp$. Believing the source would then produce trust in the information given.
%We also intend to explore further potential applications of \logic\, and of the $\supp$ operator.
In the future, we intend to explore further potential applications of \logic\ and of the support operator. From the technical point of view, we plan to study the derived operator $T$ in isolation and identify its properties independently of support and belief.
Moreover, we intend to extend \logic\ by: $(i)$ allowing nesting of the operators, using beliefs within support statements; $(ii)$ providing a proof calculus, along the line of that in~\cite{CiabattoniR23}, equipped with a prover; $(iii)$ moving towards a quantitative, dynamic, and multi-agent setting.

\bibliographystyle{plain}
\bibliography{bibliography.bib}

\appendix

\section{Technical appendix: proofs of results}\label{sec:AppendixA}

%\textbf{Theorem \ref{teo:Minimal_System}}
%The rules {\bf RW} and {\bf LLE}, as well as the axioms {\bf AND}, {\bf CUT}, and {\bf OR} are derivable in the system for $\supp$.

\begin{proof}[Theorem~\ref{teo:Minimal_System}]
The rules {\bf RW} and {\bf LLE}, as well as the axioms {\bf AND}, {\bf CUT}, and {\bf OR} are derivable in the system for $\supp$.

Trivial for \textbf{RW} and \textbf{LLE}.
%\textbf{RW} is an instance of \textbf{RCK} and \textbf{LLE} derives from \textbf{LL+}. 
   % \textbf{AND}: $((\varphi \supp \psi) \land (\varphi \supp \chi)) \rightarrow (\varphi \supp \psi \land \chi)$ \\

Axiom \textbf{AND}: 
\[\begin{array}{llr}
(1) & (\psi \land \chi) \rightarrow (\psi \land \chi) &
(\mathbf{CL}) \\
(2) & ((\varphi \supp \psi) \land (\varphi \supp \chi)) \rightarrow (\varphi \supp (\psi \land \chi)) &
(\mathbf{RCK}) \\
\end{array}
\]
%can be shown by applying \textbf{RCK} to $\psi \land \chi \rightarrow \psi \land \chi$:
%    \begin{equation*}
%            \frac{\psi \land \chi \rightarrow \psi \land \chi}{((\varphi \supp \psi) \land (\varphi \supp \chi)) \rightarrow (\varphi \supp \psi \land \chi)}
%    \end{equation*}
    %$ \neg((\varphi \supp \bot)) \rightarrow \neg (((\varphi \supp \psi)) \land ((\varphi \supp \neg \psi)))$
%The contraposition of axiom \textbf{D} $((\varphi \supp \psi) \land (\varphi \supp \neg \psi)) \rightarrow (\varphi \supp \bot) $ is an instance of \textbf{AND}. 
%$ (\varphi \supp \psi) \rightarrow (\top \supp \varphi \rightarrow \psi)$

Axiom \textbf{CUT}:
\[\begin{array}{llr}
(1) & (\varphi \supp \psi) \land ((\varphi \land \psi) \supp \chi) & \textit{Hyp.}  \\
(2) & (\varphi \supp \psi) \land (\varphi \supp (\psi \rightarrow \chi)) & (\textbf{SH}) \\
(3) & (\psi \land (\psi \rightarrow \chi)) \rightarrow \chi & (\mathbf{CL}) \\
(4) & \varphi \supp \chi &  (\mathbf{RCK}) + (2 \land 3)
\end{array}\]

% Axiom \textbf{DRDD}: 
% \[\begin{array}{llr}
% %(1) & \varphi \leftrightarrow (\top \land \varphi) & (\textbf{CL}) \\
% (1) & (\varphi \supp \chi) \leftrightarrow ((\top \land \varphi) \supp \chi ) & (\textbf{LLE}) \\
% (2) & (\varphi \supp \chi) \rightarrow (\top \supp (\varphi \rightarrow \chi)) & (\mathbf{SH}) \\
% \end{array}\]
%$((\top \land \varphi) \supp \psi) \rightarrow (\top  \supp (\varphi \rightarrow \psi))$

    % $((\varphi \supp \psi)\land (\chi  \supp \psi))\rightarrow(\varphi \lor \chi \supp \psi)$

Axiom \textbf{OR}: 

Starting from $(\varphi \land (\varphi \lor \chi)) \IFF \varphi$ (\textbf{CL}) and applying \textbf{LLE} we get $(\varphi \supp \psi) \IFF ((\varphi \land (\varphi \lor \chi)) \supp \psi$); similarly, we can get also
$(\chi \supp \psi) \IFF ((\chi \land (\chi \lor \varphi)) \supp \psi$). 
Now, let us assume by hypothesis $(\varphi \supp \psi) \land  (\chi \supp \psi)$:
\[\begin{array}{llr}
(1) & (\varphi \lor \chi) \supp (\varphi \rightarrow \psi) &  \textit{Hyp.} + (\mathbf{SH}) \\
(2) & (\varphi \lor \chi) \supp (\chi \rightarrow \psi)  
& \textit{Hyp.} + (\mathbf{SH}) \\
%(3) & ((\varphi \rightarrow \psi) \land (\chi \rightarrow \psi)) \rightarrow ((\varphi \lor \chi) \rightarrow \psi) & (\mathbf{CL}) \\
(3) & ((\varphi \rightarrow \psi) \land (\chi \rightarrow \psi)) \rightarrow ((\varphi \lor \chi) \rightarrow \psi) & (\mathbf{CL}) \\
(4) & 
(\varphi \lor \chi) \supp ((\varphi \lor \chi) \rightarrow \psi) & (\mathbf{RCK}) + (1 \land 2) \\
	% &\text{(f)} $[((\varphi \lor \chi) \supp (\varphi \rightarrow \psi)) \land ((\varphi \lor \chi) \supp (\chi \rightarrow \psi))] \rightarrow [(\varphi \lor \chi) \supp ((\varphi \lor \chi) \rightarrow \psi)]$ &(RCK) \\
(5) &
(\varphi \lor \chi) \supp ((\varphi \lor \chi) \land ((\varphi \lor \chi) \rightarrow \psi)) & (\mathbf{ID}) + (\mathbf{AND}) \\
(6) & (\varphi \lor \chi) \supp \psi & (\mathbf{RW}) \\%[-2mm]
\end{array}\]
\end{proof}

\begin{proof}[Remark~\ref{CMrejection}]
We show that \textbf{CM} in conjunction with \textbf{CUT} and the rule \textbf{RW} permits to derive \textbf{REC}:

$((\varphi \supp \psi)\land (\psi \supp\varphi))\rightarrow ((\varphi\supp\chi) \leftrightarrow  (\psi \supp\chi) )$

%as mentioned in Remark \ref{CMrejection}:

\[\begin{array}{llr}
(1) & ((\varphi \supp \psi) \land (\varphi \supp \chi)) \rightarrow ((\varphi \land \psi) \supp \chi) &
(\mathbf{CM}) \\
(2) & ((\psi \supp \varphi) \land ((\varphi \land \psi) \supp \chi)) \rightarrow (\psi \supp \chi) &
(\mathbf{CUT}) \\
(3) & ((\varphi \supp \psi) \land (\psi \supp \varphi) \land  (\varphi \supp \chi)) \rightarrow (\psi \supp \chi) &
(1 \land 2) \\
\end{array}
\]

$(3)$ is equivalent to $((\varphi \supp \psi) \land (\psi \supp \varphi)) \rightarrow ((\varphi \supp \chi) \rightarrow (\psi \supp \chi))$. By again using \textbf{CM} and \textbf{CUT} we can also derive $((\varphi \supp \psi) \land (\psi \supp \varphi)) \rightarrow ((\psi \supp \chi) \rightarrow (\varphi \supp \chi) )$ which in total gives as \textbf{REC}.

\end{proof}

% \begin{theorem} 
% \label{teo:connectionF}
%  All axioms and rules of $\F$ -- but $\mathbf{5}$ and $\mathbf{Abs}$ -- are derivable in the axiomatization for $\supp$.
% \end{theorem}

\begin{proof}[Theorem~\ref{teo:connectionF}] All axioms and rules of $\F$ -- but $\mathbf{5}$ and $\mathbf{Abs}$ -- are derivable in the axiomatization for $\supp$. \\

The claim for axioms $\mathbf{T}$, $\mathbf{Ext}$, $\mathbf{ID}$, and $\mathbf{SH}$  directly follows from the translation. For the remaining axioms: The translation of axiom \textbf{D}$^\ast$ in terms of $\supp$ is 
$\neg (\varphi \supp \bot) \to
\neg ((\varphi \supp \psi) \land (\varphi \supp \neg \psi))$. Its
contraposition $((\varphi \supp \psi) \land (\varphi \supp \neg \psi)) \rightarrow (\varphi \supp \bot) $ is an instance of \textbf{AND}.

Case $\mathbf{K_\Box}$: 
\[\begin{array}{llr}
(1) & (\neg (\varphi \rightarrow \psi) \supp \bot) \rightarrow ( \varphi \rightarrow \psi) & (\mathbf{ST}) \\
(2) & (\neg \varphi \supp \bot) \rightarrow \varphi &  (\mathbf{ST}) \\
(3) & ((\neg (\varphi \rightarrow \psi) \supp \bot) \land (\neg \varphi \supp \bot))  \rightarrow \psi & 
(\mathbf{CL}) + (1 \land 2) \\
(4) & ((\neg (\varphi \rightarrow \psi) \supp \bot) \land (\neg \varphi \supp \bot))  \rightarrow (\neg \psi \supp \bot) &  (\mathbf{S5_F}) \\
(5) & (\neg (\varphi \rightarrow \psi) \supp \bot) \rightarrow ((\neg \varphi \supp \bot) \rightarrow (\neg \psi \supp \bot)) &  (\mathbf{CL})
\end{array}\]

Case \textbf{COK}: 
\[\begin{array}{llr}
(1) & ((\varphi \supp (\psi \rightarrow \chi)) \land (\varphi \supp \psi)) \rightarrow (\varphi \supp ((\psi \rightarrow \chi) \land \psi)) &       
(\mathbf{AND}) \\
(2) & ((\varphi \supp (\psi \rightarrow \chi)) \land (\varphi \supp \psi)) \rightarrow (\varphi \supp \chi)) & (\mathbf{RW}) \\
(3) & (\varphi \supp (\psi \rightarrow \chi)) \rightarrow ((\varphi \supp \psi) \rightarrow (\varphi \supp \chi)) & (\mathbf{CL})
\end{array}\]

Case \textbf{Nec}:
\[\begin{array}{llr}
(1) & (\neg \varphi \supp \bot) \rightarrow (\varphi \lor \neg \psi) & (\mathbf{ST}) + (\mathbf{CL}) \\      
(2) & (\neg \varphi \supp \bot) \rightarrow ((\neg \varphi \land \psi) \supp \bot) & (\mathbf{S5_F}) + (\mathbf{LL+}) \\
(3) & (\neg \varphi \supp \bot) \rightarrow  (\psi \supp \varphi) & 
(\mathbf{SH}) + (\mathbf{CL})
\end{array}\]

Case Necessitation for $\Box$, if we assume that the formula $\varphi$ is provable, then we can derive $\neg \varphi \supp \bot$ via the following:
%\\ \textbf{N}: 
\[\begin{array}{llr}
(1) & \neg \varphi \rightarrow \bot & \textit{Hyp.} + (\mathbf{CL}) \\
(2) & (\neg \varphi \supp \neg \varphi) \rightarrow (\neg \varphi \supp \bot) & (\mathbf{RCK}) \\
(3) & \neg \varphi \supp \neg \varphi  & (\mathbf{ID}) \\
(4) & \neg \varphi \supp \bot & (2 \land 3) \\[-2mm]
\end{array}\]

\end{proof}

\begin{proof}[Lemma~\ref{lem:truth-lemma} (Truth lemma)]
   $\mathcal{M}^{Can}, (\Delta, \pi, i) \models \alpha$ iff $\alpha \in \Delta$.

This proof proceeds by structural induction on the formula $\alpha$. 
   
If $\alpha \in \ForCL$ the claim follows directly from $\mathbf{CL}$. Assuming the induction hypothesis for every formula in $\ForCL$ lets us derive $\forall \varphi \in \ForCL$ $||\varphi||_\Gamma=[\varphi]_\Gamma$ and $\mathit{most}(||\varphi||_\Gamma)=\mathit{max}([\varphi]_\Gamma)$. 

Let $\alpha$ be of the form $\varphi \supp \psi$. We start with the case $\alpha \in \Delta$. Take $\Gamma$ with $\Delta \in [\Gamma]_\leftrightsquigarrow$. Given an arbitrary $(\Omega, \gamma, j) \in S_\Gamma$ with $(\Omega, \gamma, j) \in \mathit{max}(||\varphi||_\Gamma)$ we get $(\Omega, \gamma, j) \in \mathit{max}([\varphi]_\Gamma)$ by the induction hypothesis. Because $\psi \in \supp_\varphi(\Gamma)$ Corollary~\ref{cor:iff} implies that $\psi \in \Omega$. Again by using the induction hypothesis we conclude  $(\Omega, \gamma, j) \models \psi$ which means $(\Delta, \pi, i) \models \varphi \supp \psi$.

For the other direction, we assume $\alpha \not\in \Delta$. In this case, we have to find a maximal $\varphi$ state which does not satisfy $\psi$. The assumption $\psi \not\in \supp_\varphi(\Delta)$ lets us derive $\psi \not\in \supp_\varphi(\Gamma)$. By the use of Corollary~\ref{cor:iff}, we obtain a $(\Pi, \chi, i) \in \mathit{max}([\varphi]_\Gamma)$ with $\psi \not\in \Pi$. By induction hypothesis we get  $(\Pi, \chi, i) \in \mathit{max}(||\varphi||_\Gamma)$ and $(\Pi, \chi, i) \models \neg \psi$. This means $(\Delta, \pi, i) \not\models \varphi \supp \psi$. 

Let $\alpha$ be of the form $B(\varphi)$. Then
both directions of the claim follow directly from the induction hypothesis and the construction of $R$ in Definition~\ref{CanMod}.
\end{proof}

\end{document}